\documentclass[12pt,final,onecolumn]{IEEEtran}
\usepackage{amsmath,amssymb,amsthm,amsfonts,mathtools,enumitem,algorithmic}
\usepackage[colorlinks=true, linkcolor=black, filecolor=black, urlcolor=black,citecolor=black]{hyperref}
\usepackage{graphicx,float}
\usepackage{caption}
\usepackage{subcaption}
\theoremstyle{theorem}
\newtheorem{theorem}{Theorem} 
 
\theoremstyle{definition}
\newtheorem{definition}{Definition}

\newtheorem{example}{Example}
\theoremstyle{remark}
\newtheorem{remark}{Remark}
\newcommand{\diag}[1]{\textrm{diag}\{#1\}}
\newcommand{\col}[1]{\textrm{col}\{#1\}}
\newcommand{\idm}[1]{\mathbf{I}_{#1}}
\newcommand{\zdm}[1]{\mathbf{0}_{#1}}
\newcommand{\odm}[1]{\mathbf{1}_{#1}}
\DeclareMathOperator*{\argmin}{arg\,min}

\newcommand{\tb}[1]{\textcolor{black}{#1}}
\def\BibTeX{{\rm B\kern-.05em{\sc i\kern-.025em b}\kern-.08em
		T\kern-.1667em\lower.7ex\hbox{E}\kern-.125emX}} 
\begin{document}
	\title{Dynamic network analysis of a target defense differential game with limited observations}
	\author{Sharad Kumar Singh and Puduru Viswanadha Reddy, \IEEEmembership{Member, IEEE
		} 
		\thanks{S.K.Singh and P.V.Reddy are with the Department of Electrical Engineering, Indian Institute of Technology, Madras, Chennai, India. email:  {~ee15d201@smail.iitm.ac.in,~ vishwa@ee.iitm.ac.in}}}

	\maketitle

	\begin{abstract}  In this paper, we study a Target-Attacker-Defender (TAD) differential game involving one attacker, one target and multiple defenders. We consider two variations where (a) the attacker and the target have unlimited observation range and the defenders are visibility constrained
		(b) only the attacker has unlimited observation range  and the remaining players are visibility constrained.  We model the players' interactions as a dynamic game with asymmetric information. Here, the visibility constraints of the players induce a visibility network which encapsulates the visibility information during the evolution of the game. Based on this observation, we introduce network adapted feedback or implementable strategies for visibility constrained players. Using inverse game theory approach  we obtain network adapted feedback Nash equilibrium strategies. 
		We introduce  a consistency criterion for selecting a subset (or refinement) of network adapted feedback Nash strategies, and provide an optimization based approach for computing them.  Finally, we illustrate our results with numerical experiments.
	\end{abstract}
	

	\section{Introduction} 
	{The study of networked autonomous multi-agent interactions has received a lot of interest, in   recent years, in the areas such as  surveillance, rescue and combat missions,   navigation, and analysis of biological behaviors; see \cite{GHu:2017,Gupta:2017,Nguyen:2017,Lanzon:2020,Huang:2020}}.  This paper 
	is motivated by strategic situations observed in the practical engineering applications such as, protection of critical infrastructures (e.g., aircrafts, naval ships, power grid) against attacks from incoming missiles, interceptors defending an asset against intrusions, and biological interactions involving protection of the young from predators.
	
	The above situations are  analyzed using the mathematical framework of pursuit-evasion games \cite{isaacs65a} with two or three players (groups).
	A {two-player} (group) interaction is referred to as a pursuit evasion (PE) game.
	Here, the objective of the pursuer is to capture the evader which tries to avoid being captured by the pursuer. A {three-player} (group) interaction is referred to as 
	a Target-Attacker-Defender (TAD) game  \cite{Boyell:1976a}, \cite{Boyell:1980a}. Here, the goal of the attacker is to \emph{capture} the target which tries to evade the attacker, and the goal of the defender is to \emph{intercept} the attacker before the attacker captures the target. {Three-player} interactions resulting in the \emph{rescue} of the target by the defender  have been studied recently \cite{oyler2016pursuit}. Clearly, a TAD game is far more complex than a PE game in that the former involves two simultaneous PE type interactions resulting in more outcomes. 
	\subsection{Contributions}    The existing literature on TAD games assume that all the players have unlimited sensing capabilities which allow them to observe other players during their interactions. However, in the real-world, a player (an engineered agent) has limited sensing capabilities, and can observe other players only when they
	are within its sensing range. For example, this situation occurs when a team of unmanned ground vehicles (UGVs), equipped with inferior sensing capabilities, must safeguard an asset from potential attacks by a well equipped UGV with superior sensing capabilities. Further, limited sensing situations can also arise due to potential failures in the sensing equipment during interactions.   
	
	The novelty of this paper lies in the study of TAD games that involve  players with limited visibility capabilities.  In many real-world applications in civilian or military settings, limited visibility is  an important challenge to address, and has practical implications leading to questions such as  (i) how would players adapt their strategies to the evolving visibility information during their interactions? (ii) under what informational assumptions can the players synthesize their implementable strategies? This paper aims to address these questions by developing a general framework for analyzing TAD interactions with limited observations.  We use differential game methodology, more specifically the  Game of  Degree approach, for modeling interactions among the players; see  \cite{isaacs65a}, \cite{li2011defending}. We note that players' visibility constraints induce a (dynamic) directed network which captures the evolution of visibility information. To be deemed implementable,  the strategies of the  players must be adapted to this information. 
	
	To  address question (i), we introduce the notion of network adapted feedback strategies or the implementable strategies for the visibility constrained players.   A TAD game with limited observations can have many possible interaction configurations.  In this paper, we focus on two scenarios. In the first, we assume that the attacker and the target have  unlimited visibility range. We assume that the (multiple) defenders have limited visibility range, and due to which they act as a team. We model this interaction as a non-zero-sum linear quadratic differential game. In the second, we assume that both the target and the defenders are visibility constrained, due to which they act as a team against the attacker who has an unlimited visibility range. We model this interaction as a zero-sum linear quadratic differential game. We emphasize that our choice of scenarios is canonical, leading to non-zero-sum and zero-sum dynamic game models, and using the framework developed in this paper, the other interaction configurations can be studied. 
	
	To address question (ii), we assume that players use feedback Nash equilibrium strategies 
	as an outcome of their interactions. Since the visibility constrained players cannot have complete observations, we synthesize their network adapted feedback Nash strategies  using an inverse game theory approach, in Theorems \ref{thm:thm1} and \ref{thm:thm2}. Being an inverse problem, we obtain a plethora of implementable strategies. Then, based on the idea that when all the players can observe others, their implementable strategies must be same as their standard feedback Nash strategies, we develop an information consistency criterion for selecting the implementable strategies. In Theorem \ref{thm:consistency}, we provide an optimization based approach for synthesizing the implementable strategies. Further, in Theorem \ref{thm:delayproperty}, we perform sensitivity analysis of visibility radii to analyze the effect of visibility parameters on these strategies. 
	
	%%%
	The paper is organized as follows. Preliminaries and problem formulation are presented in  section \ref{sec:problem_formulation}.   We analyze the first variation of the TAD game in section \ref{sec:diffgame}, and the second variation in \ref{sec:feedback_zsg}. In section \ref{sec:optimization}, we introduce information consistency based procedure for selecting the feedback gains of the visibility constrained players. In section \ref{sec:simulations}, we illustrate our results with numerical simulations. Section \ref{sec:conclusions} provides concluding remarks and a summary of future research.
	%%%
	%%%
	%%%
	\subsection{An overview of related  literature }
	\label{sec:literature} 
	TAD type interactions were studied in \cite{Boyell:1976a,Boyell:1980a} in the context of defending ships from an incoming torpedo using counter-weapons. 
	A TAD type interaction referred to as the lady, the bandits and the body-guards was proposed in \cite{rusnak2005lady}.
	In \cite{Rubinsky:2014three}, the authors study a TAD terminal game and propose attacker strategies for evading the defender while continuing to pursue the target. In \cite{li2011defending}, the authors study the problem of defending an asset, by modeling the interactions as a linear quadratic differential game, and proposed moving horizon strategies for different configurations of the target. In \cite{Venkatesan:2015new}, a guidance law for defending a non-maneuverable aircraft is proposed.   In \cite{Ratnoo:2011line}, \cite{Ratnoo:2012guidance}, \cite{Shima:2011optimal}, \cite{Prokopov}, \cite{Perelman}, \cite{Shaferman}, the authors study the problem of defending aircrafts from an incoming homing missile using defensive missiles by considering various interaction scenarios.   	
	{In  \cite{Garcia:2017c,Garcia:2018a,Garcia:2019, Pachter:2019toward},} the authors  study interactions where a homing missile tries to pursue an aircraft, and a defender missile aims at intercepting the attacker. In particular, 
	they study 	cooperative mechanisms between the target-defender team against the attacker so that the defender can intercept the attacker before the attacker can capture the target.   	Role switching  of attacker in TAD games was studied recently in \cite{Liang:2020}. In  \cite{Fuchs:2016generalized}, the defender's strategies force the attacker to retreat instead of engaging the target. In \cite{Weiss:2017combined}, the authors study the possibility of role switch as well as the cooperation between the target and defender. The recent tutorial article \cite{TutorialPETAD:2020} provides a survey of PE and TAD differential games.  {In all the above TAD game related works}, all the  players are assumed to have unlimited observations without visibility constraints.
	
	In the context of PE games, {\cite{Bopardikar:2008}, \cite{LaValle:2001}, \cite{lin2015nash},} analyze interactions involving players with limited sensing capabilities. In \cite{lin2015nash}, the authors study a PE interaction between one evader with unlimited observation range and multiple pursuers with limited visibility capabilities. This is the closest reference we could find related to our work. {This paper  differs from \cite{lin2015nash} both in scope} and content as a TAD game is far more complex than a PE game. In particular, the differences with \cite{lin2015nash} are as follows. In our work, the structure of network feedback adaptive strategies is provided in a very general setting, in that the feedback gains
	matrices associated with the neighboring defenders are different, {whereas in \cite{lin2015nash},}
	all the feedback gain matrices of the neighboring pursuers are the same. In a TAD game, different interaction configurations can arise among the players, and we have considered these possibilities in our work in sections \ref{sec:feedback_nzsg} and \ref{sec:feedback_zsg}. {Further, we develop the notion of information consistency, which was not studied before in the literature, towards  the refinement of network adapted feedback Nash strategies.} 
	\vskip1ex 
	\noindent 
	\textbf{Notation}:  Throughout this paper, $\mathbb R^n$ denotes the set of $n$- dimensional real  column vectors, and $\mathbb R^{n\times m}$
	denotes the set of $n\times m$ real matrices. The symbol $\otimes$ denotes the Kronecker product. The transpose of a vector or matrix $E$ is denoted by $E^\prime$.  The Euclidean norm of a vector $x\in \mathbb R^n$ is denoted by $||x||_2= \sqrt{x^\prime x}$. For any $x\in \mathbb R^n$ and $ S\in \mathbb R^{n\times n}$, we denote the quadratic term $x^\prime S x$ by $||x||_S^2$, and the Frobenius norm of $S$ by $||S||_f=\sqrt{\text{trace}(S^\prime S)}$. $\mathbf{0}_{m\times n}$ denotes the $m\times n$ matrix with all its entries equal to zero. $\mathbf{I}_n$ denotes identity matrix of size $n$, and $\mathbf{e}_n^i$ denotes the $i$th column of $\mathbf{I}_n$.  $\mathbf{1}_n$ denotes the $n\times 1$ vector with all its entries equal to $1$.
	$\col{e_1,\cdots,e_n}$ denotes the single vector or matrix obtained by stacking
	the vectors or matrices $e_1,\cdots,e_n$ vertically. $\diag{e_1,\cdots,e_n}$
	denotes the block diagonal matrix obtained by  taking the matrices $e_1,\cdots,e_n$ as diagonal elements in this sequence.  	
	A directed network (or a graph) is denoted by a pair $\mathcal G:=(\mathcal V, \mathcal E)$. $\mathcal V=\{v_1,\cdots,v_n\}$ denotes the set of vertices, and
	$\mathcal E\subseteq \{(v_i,v_j)\in \mathcal V\times \mathcal V~|~v_i,v_j \in \mathcal V,~i\neq j\}$ denotes the set of directed edges without self-loops. A directed edge 	of $\mathcal G$, from $v_i$ to $v_j$, is denoted by $v_i\rightarrow v_j:=(v_i,v_j)$.
	
		\section{Preliminaries and Problem formulation}\label{sec:problem_formulation}
	\subsection{Dynamics and interactions of the players}
	We consider a team of $n$ defenders denoted by $\mathcal{D}:=\{d_1, d_2, \cdots, d_n\}$, the target by $\tau$, and  the attacker by $a$. The set of players is denoted by $\mathcal P :=\mathcal D~\cup~ \{\tau ,a\}$. We assume that the players interact in a two-dimensional plane. The dynamics of each player is governed by the following single integrator dynamics
	\begin{align}	  {
		\begin{bmatrix} \dot{x}_p(t)& \dot{y}_p(t)\end{bmatrix}^\prime= 
		\begin{bmatrix} u_{px}(t) &u_{py}(t)\end{bmatrix}^\prime,}  
	\end{align}
	where $(x_p(t),y_p(t))\in \mathbb{R}^2$ is the position vector of the player $p\in \mathcal P$ at time $t$, $(u_{px}(t),u_{py}(t))\in \mathbb{R}^2$ represents the  control input  of player $p$ at time $t$, and $(x_{p0},y_{p0})\in \mathbb{R}^2$ represents the initial position vector of player $p$. We denote the state and control vector of player $p\in \mathcal P$ as
	\begin{equation}\label{eq:position}
	 {X_p(t)=\begin{bmatrix}
		x_p(t)&	y_p(t)
		\end{bmatrix}^\prime,~~u_p(t)=\begin{bmatrix}u_{px}(t)&
		u_{py}(t)
		\end{bmatrix}^\prime.}
	\end{equation}
	 {In this paper, we study the following two variations of TAD-type interactions among the players.}		 
	 {
		\begin{enumerate}[label=\textbf{I\arabic*.},wide, labelindent=0pt] 
			\item   \begin{enumerate} [labelindent=0pt,wide]
				\item The attacker and target have unlimited visibility range. The visibility constrained defenders cooperate as a team  and try to intercept the attacker. A \textit{non-suicidal} attacker tries to evade the defenders while trying to capture the target, whereas a \textit{suicidal} attacker--not interested in its survival--only tries to capture the target. The target always tries to evade the attacker.
				\item The attacker and target are not aware of defenders' visibility constraints. All the defenders are aware of their own visibility constraints as well as unlimited visibility capabilities of the attacker and target. Further, defenders also know that the attacker and target are unaware of defenders'  visibility constraints. 
			\end{enumerate}
			\item  \begin{enumerate}[labelindent=0pt,wide]
				\item The attacker has an unlimited observation range. The target and defenders have visibility constraints and cooperate as a team.  The defenders try to intercept the attacker before it captures the target, which  always tries to evade the attacker.
				\item Same as in \textbf{I1}.b) by replacing the attacker and target with the attacker, and the defenders with defenders-target team.			 		 
			\end{enumerate} 
	\end{enumerate} }
	Besides consideration of multiple defenders, our work distinguishes from the existing literature due to the following features in the interactions.
	\begin{enumerate} [label=\textbf{F\arabic*}.,wide,labelindent=0pt]
		\item  {The visibility constrained players act as a team in achieving their objectives}. The limited observation range of the   {visibility constrained players} induce  a time-varying (directed) visibility network which captures information of the state variables that are accessible to the  {visibility constrained players} during the interactions.
		\item Whenever the  {visibility constrained players}  cannot see other players, the interaction results in a situation where the available information  is asymmetric, that is,   {visibility constrained players} have private information about their visibility capabilities.
	\end{enumerate}
	 {Due to the nature of interactions, the state space can be reduced using relative coordinates. That is, at any time instant $t$, we denote by $z_{p}(t)=X_{p}(t)-X_a(t)$  the displacement vector between the player $p$ ($p\in \mathcal P\backslash \{a\}$) and the attacker $a$. The global state vector associated with the reduced
		state space be denoted by $z(t)=\col{z_{d_1}(t),\cdots, z_{d_n}(t),z_\tau(t)}\in \mathbb R^{2(n+1)}$. Using this, the dynamic interaction environment of the players can be written compactly as
		\begin{align}\label{eq:stateequation}
		\dot{z}(t)=\sum_iB_{d_i}u_{d_i}(t)+B_\tau u_\tau(t)+B_au_a(t),
		\end{align}
		where   $B_{d_i}=\begin{bmatrix}{\mathbf{e}_n^i}^\prime & 0 \end{bmatrix}^\prime \otimes \idm{2}$, $B_\tau =\begin{bmatrix} \zdm{1\times n} & 1\end{bmatrix}^\prime \otimes \idm{2}$, and $B_a =-\odm{n+1} \otimes \idm{2}$.}
	
	\subsection{Network induced by visibility constraints}
	\label{sec:network}
	 {We assume that a visibility constrained player $p$ ($p\in \mathcal P\backslash \{a\})$ can see a player $q\in \mathcal P \backslash \{p\}$, at time $t$, when the player $q$ lies within player $p$'s observation radius $\zeta_{p}>0$, that is,  when the following condition holds true
		\begin{align} 
		||X_p(t)-X_q(t)||_2=||z_p(t)-z_q(t)||_2\leq \zeta_{p}. \label{eq:limob}
		\end{align}
		We set  $\zeta_\tau=\infty$ ($\zeta_\tau<\infty$) as target has unlimited  (limited) visibility range in the interaction \textbf{I1} (\textbf{I2})
		The above constraint induces a time-varying directed  network $\mathcal G(t):=(\mathcal P, \mathcal E(t))$, where an outgoing edge $p\rightarrow q\in \mathcal E(t)$ indicates that  a player $q\in \mathcal P\backslash \{p\}$ is visible to the visibility constrained player $p$ at time $t$. Fig. \ref{fig:playerINTgraph} illustrates the visibility network associated with the interactions given in Fig. \ref{fig:playerINT}}. 

		 \begin{figure}[h] 
		 	\begin{subfigure}[b]{.5\textwidth}\centering 
		 		\includegraphics[scale=1.25]{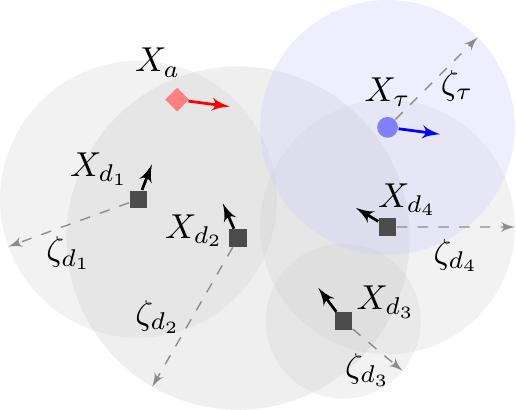}
		 		\caption{Interactions between the players} 
		 		\label{fig:playerINT} 
		 	\end{subfigure}	 
		 	\begin{subfigure}[b]{.5\textwidth}\centering 
		 		\includegraphics[scale=1.125]{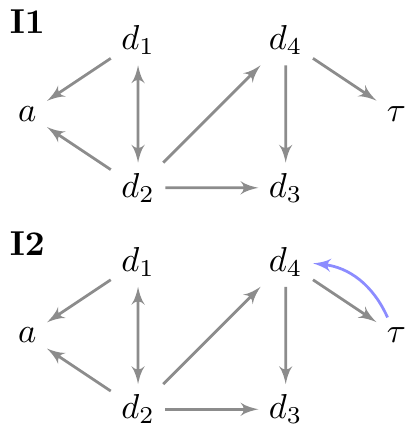}
		 		\caption{Directed network induced by the visibility information of the players} 
		 		\label{fig:playerINTgraph} 
		 	\end{subfigure}
		 	\caption{In panel (a) the limited observation range of the defenders (target) is illustrated by the light gray (blue) regions. In panel (b).I1   the outgoing edges from the attacker and the target are omitted they can see all the players in interaction \textbf{I1} ($\zeta_\tau=\infty$). Similarly, in panel (b).I2 outgoing edges from the attacker are suppressed as it can see all the players in the interaction \textbf{I2} ($\zeta_\tau<\infty$).}  
		 	\label{fig:int_game}
		 \end{figure}

  \subsection{Termination criterion}
 Let the positive real numbers $\sigma_{d_i}<\zeta_{d_i}$ ($d_i\in \mathcal D$) and $\sigma_a$ denote the capture radii of the defender $d_i$ and the attacker $a$ respectively. The interactions terminate when a defender intercepts the attacker, that is,  whenever $||X_{d_i}(t)-X_a(t)||_2\leq \sigma_{d_i}$ holds for at least one $d_i\in \mathcal D$, or when the attacker captures the target, that is, whenever $||X_a(t)-X_\tau(t)||_2\leq \sigma_a$ holds.   
 
 \subsection{Problem statement} {Let $T<\infty$ denotes the duration of the interaction. 
 	We assume that the duration of the interaction is large enough so that the termination criterion is satisfied during the interval $ [0,T]$. We seek to determine the control strategies which can be used by the players during the interactions \textbf{I1} and \textbf{I2}.}  {In the next two sections, using the differential game approach \cite{Basar:99}, we analyze the interactions  \textbf{I1} and \textbf{I2}.}
		  \section{ {Analysis of interaction \textbf{I1}}} 
		  \label{sec:diffgame}
		   {Recall that the visibility constrained defenders act as a team in the interaction \textbf{I1}.  We denote  the control input for team  of defenders be denoted by $u_d(t)=\col{u_{d_1}(t), \cdots, u_{d_n}(t)} \in \mathbb R^{2n}$. Using this, the dynamic interaction environment of the players  is written as
		  \begin{align}\label{eq:state_eq}
		  \dot{z}(t)=B_{d}u_{d}(t)+B_\tau u_\tau(t)+B_au_a(t),
		  \end{align}
		  where   $B_{d}=\begin{bmatrix}\idm{n} & \zdm{n\times 1} \end{bmatrix}^\prime \otimes \idm{2}$.} 		  The objectives of the players are described as follows. 
		  The target maximizes its weighted distance with the attacker during the time period $[0,T]$. The defenders jointly minimize the sum of their individual weighted distances with the attacker during the interval $[0,T]$. The objective of the non-suicidal attacker is to maximize the sum of its weighted distances with the defenders while simultaneity minimizing its weighted distance with the target during the time period $[0,T]$. Whereas a suicidal attacker minimizes its distance with the target during the time period $[0,T]$. All the players simultaneously minimize their control efforts during the time period $[0,T]$. Using their    controls $u_d(.)$,  defenders jointly minimize the following objective function
		  subject to \eqref{eq:state_eq}
\begin{subequations} 
	\begin{align} 
J_d(u_d(.), u_\tau(.),u_a(.))&=\frac{1}{2} \sum_{j=1}^{n} ||z_{d_j}(T)||^2_{F_{d_ja}}
	 + \frac{1}{2} \int_{0}^{T} \Big( \sum_{j=1}^{n} ||z_{d_j}(t)||^2_{Q_{d_ja}} + ||u_{d_j}(t)||^2_{R_{d_j}}  \Big)dt \nonumber\\
	& =\frac{1}{2} ||z(T)||^2_{F_{d}} + \frac{1}{2} \int_{0}^{T} \Big( ||z(t)||^2_{Q_{d}} + ||u_{d}(t)||^2_{R_{d}} \Big)dt. \label{eq:obj:def}   
	\end{align}
	Using the controls $u_\tau(.)$, the target minimizes the following objective function subject to \eqref{eq:state_eq}
	\begin{align} 
 J_\tau(u_d(.), u_\tau(.),u_a(.)) &=-\frac{1}{2} ||z_{\tau}(T)||^2_{F_{\tau a}} + \frac{1}{2} \int_{0}^{T}  \Big(  ||u_{\tau}(t)||^2_{R_{\tau}}-||z_{\tau}(t)||^2_{Q_{\tau a}}   \Big)dt\nonumber\\
	&=\frac{1}{2} ||z(T)||^2_{F_\tau} + \frac{1}{2} \int_{0}^{T} \Big( ||z(t)||^2_{Q_{\tau}}+ ||u_{\tau}(t)||^2_{R_{\tau}} \Big)   dt. \label{eq:obj:tar}
	\end{align} 
	Using the controls  $u_a(.)$, the attacker minimizes the following objective function subject to \eqref{eq:state_eq}
	\begin{align}
	&J_a(u_d(.), u_\tau(.),u_a(.))=
	\frac{1}{2} ||z_{\tau}(T)||^2_{F_{a\tau}}  -\frac{\lambda}{2}\sum_{j=1}^{n}||z_{d_j}(T)||^2_{F_{ad_j}} \nonumber \\     &+ \frac{1}{2} \int_{0}^{T} \Big( ||z_{t}(t)||^2_{Q_{a\tau}}-\lambda \sum_{j=1}^{n}||z_{d_j}(t)||^2_{Q_{ad_j}}+ ||u_{a}(t)||^2_{R_{a}}  \Big)dt\nonumber\\
	&=\frac{1}{2} ||z(T)||^2_{F_{a}} + \frac{1}{2} \int_{0}^{T} \Big( ||z(t)||^2_{Q_{a}} + ||u_{a}(t)||^2_{R_{a}}  \Big)dt. \label{eq:obj:atta}
	\end{align} 
	\label{eq:fullobj}
\end{subequations} 
		  where $F_{ap}=f_{ap}\idm{2}$, $F_{pa}=f_{pa}\idm{2}$, $Q_{ap}=q_{ap}\idm{2}$, $Q_{pa}=q_{pa}\idm{2}$, $p\in \mathcal D\cup \{\tau\}$ and  {$R_p=r_p \idm{2}$, $p\in  \{\tau,a \}$} with $f_{ap}>0$, $f_{pa}>0$, $q_{ap}>0$, $q_{pa}>0$ and $r_p>0$. Using these, the matrices associated with the terminal costs are given by $F_d=\diag{f_{d_1 a},\cdots,f_{d_na},0}\otimes \idm{2}$,
		  $F_\tau=\diag{\zdm{2n\times 2n},-f_{\tau a}\idm{2}}$, and $F_a=\diag{-\lambda f_{ad_1 },\cdots,-\lambda f_{ad_n},f_{a\tau}}\otimes \idm{2}$.
		  The matrices associated with the instantaneous costs are given by 
		  $Q_d=\diag{q_{d_1 a},\cdots,q_{d_na},0}\otimes \idm{2}$, 
		  $Q_\tau=\diag{\zdm{2n\times 2n},-q_{\tau a}\idm{2}}$, and $Q_a=\diag{-\lambda q_{ad_1 },\cdots,-\lambda q_{ad_n},q_{a\tau}}\otimes \idm{2}$.
		  Finally, the control cost parameter $R_d$ is given by 
		  $R_d=\diag{r_{d_1},\cdots , r_{d_n}}\otimes \idm{2}$. The parameter $\lambda\in \{0,1\}$ in \eqref{eq:obj:atta}   is set to $\lambda=1$ ($\lambda=0$) for a non-suicidal (suicidal) attacker. 
		  
		  As reflected in the players’ objectives, both conflict and cooperation co-exist from the strategic interaction of players, and necessitates the analysis using non-zero sum differential games  \cite{Basar:99}.  {In particular, the state dynamics and objectives of the players defined in \eqref{eq:state_eq}-\eqref{eq:fullobj} constitute a non-zero-sum linear quadratic differential game (NZLQDG)}; see \cite{Basar:99} and \cite{engwerda2005lq}. In a differential game, the strategies or the controls used by players depend upon the information available to them during the game, also  referred to as information structure. In the  feedback information structure, the control of a player $p\in \mathcal P$ at time $t\in[0,T]$ is a function of time $t$ and state variable $z(t)$, that is, $u_p(t)=\gamma_p(t,z(t))$, where the mapping $\gamma_p:[0,T]\times \mathbb{R}^{2(n+1)}\rightarrow \mathbb{R}^2$ is a strategy of the player $p$. In this paper, we assume  feedback strategies due to their robustness towards perturbations in the state variable, also referred to as strong time consistency property; see \cite[definition 5.14]{Basar:99}. Due to linear dynamics \eqref{eq:state_eq} and quadratic objectives \eqref{eq:fullobj}   we restrict to linear feedback strategies, and the set of feedback strategies of a player $p\in \mathcal P$ is  given by 
		  \begin{multline}
		  \Gamma_p:=\Bigl\{ \gamma_p:[0,T]\times \mathbb R^{2(n+1)} \rightarrow \mathbb R^2 ~\Bigl|~ u_p(t)=\gamma_p(t,z(t))=G_p(t)z(t),\\G_p(t)\in \mathbb R^{2\times 2(n+1)}~t\in[0,T]\Bigr\}.
		  \label{eq:fbstatset}
		  \end{multline}    
		  Since the defenders act as a team, we denote by $\gamma_d=\col{\gamma_{d_1}, \cdots,\gamma_{d_n}}$ and 
		  $\Gamma_d=\Gamma_{d_1}\times  \cdots \times \Gamma_{d_n}$
		  as their joint feedback strategy  and joint feedback strategy set respectively.

		  \subsection{ {Network adapted feedback information structure}}\label{sec:feedback_nzsg}
		   {In this subsection} we motivate the need for modifying the information structure of the defenders due to their visibility constraints. 	As the attacker, target, and the team of defenders  individually minimize their interrelated objectives, Nash equilibrium \cite{Basar:99} is a natural choice for the outcome of  {NZLQDG}.   The Nash equilibrium in feedback strategies is defined as follows.
		  \begin{definition}[Nash equilibrium]
		  	The strategy profile $(\gamma^*_{d}, \gamma^*_{\tau},  \gamma^*_{a})$ is a feedback Nash equilibrium  (FNE) for  {NZLQDG} if the following set of inequalities hold true
		  	\begin{subequations} 
		  		\begin{align}
		  		&J_d(\gamma^*_d,\gamma^*_\tau,\gamma^*_a) \leq J_d(\gamma_d,\gamma^*_\tau,\gamma^*_a), ~\forall \gamma_d\in  \Gamma_d,\\
		  		&J_\tau(\gamma^*_d,\gamma^*_\tau,\gamma^*_a) \leq J_\tau(\gamma_d^*,\gamma_\tau,\gamma^*_a), ~\forall \gamma_\tau\in  \Gamma_\tau,\\
		  		&J_a(\gamma^*_d,\gamma^*_\tau,\gamma^*_a) \leq J_a(\gamma^*_d,\gamma^*_\tau,\gamma_a), ~\forall \gamma_a\in  \Gamma_a.
		  		\end{align}
		  	\end{subequations} 
		  \end{definition}
		  The following theorem from \cite{engwerda2005lq} characterizes FNE.
		  \begin{theorem}{\cite[Theorem 8.3]{engwerda2005lq}} Consider the $(n+2)$--player  finite horizon  {NZLQDG} described by \eqref{eq:state_eq}-\eqref{eq:fullobj}. This game has,  for every initial state, a linear FNE if and only if the following set of coupled Riccati differential equations (RDE) has a set of symmetric solutions $\{P_d(t), P_\tau(t), P_a(t)\}$ on $[0,T]$
		  	\begin{subequations}
		  		\begin{align}
		  		\dot{P_d}(t)&=P_d(t) S_d P_d(t)+P_d(t) S_\tau P_\tau (t)+P_d (t)S_a P_a(t)\notag\\
		  		&\quad +P_\tau (t) S_\tau P_d(t)+P_a(t) S_a P_d (t)-{Q}_d, \label{eq:fbne1}\\
		  		\dot{P_\tau }(t)&=P_\tau (t) S_d P_d(t)+P_\tau (t) S_\tau P_\tau (t)+P_\tau (t) S_a P_a(t) \notag\\&\quad +P_d(t) S_d P_\tau (t) 
		  		+P_a(t) S_a P_\tau (t)-Q_t, \label{eq:fbne2}\\
		  		\dot{P_a}(t)&=P_a(t) S_d P_d(t)+P_a(t) S_\tau P_\tau (t)+P_a(t) S_a P_a(t)\notag\\ 
		  		&\quad		+P_d(t) S_d P_a(t)+P_\tau (t) S_\tau P_a(t)-Q_a \label{eq:fbne3}
		  		\end{align}
		  		\label{eq:Riccati}
		  	\end{subequations}
		  	with $P_d(T)= {F}_{d}$, $P_\tau(T)= {F}_{\tau}$ and $P_a(T)= {F}_{a}$. 
		  	Moreover, in that case there is a unique equilibrium. The FNE control actions of player $p \in \{d,\tau,a\}$ is given by
		  	\begin{align}
		  	\label{eq:1d}
		  	u^*_{p}(t)=\gamma_p^*(t,z(t))=-R_{p}^{-1}B_p^\prime  P_p(t)z(t).
		  	\end{align}
		  	\label{thm:fbne}
		  \end{theorem}
		  %%%
		  Due to their unlimited observation range, the attacker and target can implement their FNE strategies \eqref{eq:1d} as they have access to the state information $z(t)$ for all $t\in[0,T]$. The FNE strategy of the defender $d_i$ $(d_i\in \mathcal D)$ can be rewritten as 
		 \begin{align}
		 u_{d_i}^*(t)=\left( {\mathbf{e}_n^i}^\prime \otimes \idm{2} \right)u_d^*(t)   &=-r_{d_i}^{-1}\Bigl[\sum_{j=1}^{n+1}  {P_d^{ij}(t)}\Bigr] z_{d_i}(t) 
		 +r_{d_i}^{-1}  {P_d^{i (n+1)}(t)}\left(z_{d_i}(t)-z_\tau(t)\right)
		 \notag\\
		 &\quad +r_{d_i}^{-1}\sum_{j=1}^{n} {P_d^{ij}(t)}\left( z_{d_i}(t)-z_{d_j}(t)\right),
		 \label{eq:conimpl}
		 \end{align}	 
		  where matrix $P^{ij}_d(t)$ is the $i$th row and $j$th column element (a block matrix)
		  obtained by partitioning the matrix $P_d(t)$ into block matrices of dimension $2\times 2$. 
		  \begin{remark}
		  	From \eqref{eq:conimpl}, it is clear that strategy of the defender $d_i$
		  	not only depends on the visibility of the attacker and target, but also on the visibility of the other defenders  {$\mathcal D\backslash \{d_i\}$}. This implies, the FNE strategy \eqref{eq:1d} is not implementable by the defenders under limited observations.
		  	More specifically, at time $t$, if a player  {$p\in \mathcal P\backslash \{d_i\}$} lies outside the visibility range of the defender $d_i$ then the coefficient of the term $X_{p}(t)-X_{d_i}(t)= z_p(t)-z_{d_i}(t)$ in \eqref{eq:conimpl} must be zero for the defender $d_i$ to implement the FNE control \eqref{eq:conimpl} at time $t$.
		  	In other words, the feedback strategies \eqref{eq:fbstatset} of the defenders  must be \textit{adapted} to visibility  network $\mathcal G(t):=(\mathcal P, \mathcal E(t))$, induced by \eqref{eq:limob} at every time instant $t\in [0,T]$, to be deemed implementable.
		  \end{remark} 
		   {We recall  that at any time $t\in[0,T]$ an outgoing edge $d_i\rightarrow p\in \mathcal E(t)$ in the visibility network  $\mathcal G(t):=(\mathcal P, \mathcal E(t))$ 
		  indicates that the defender $d_i$ can see the player $p\in \mathcal P\backslash \{d_i\}$ at time $t$, whenever $||X_{d_i}(t)-X_{p}(t)||_2=||z_{d_i}(t)-z_{p}(t) ||_2\leq \zeta_{d_i}$. 
		  Using this, for the defender $d_i$ ($d_i\in \mathcal D$) to indicate the visibility of a player $q$ ($q\in \{a,\tau\}$) at time $t\in[0,T]$, we define the binary function $\phi^q_{d_i}:[0,T]\rightarrow \mathbb R$ ($q\in \{a,\tau\}$)  as
		  \begin{equation}
		  \phi^q_{d_i}(t)=\begin{cases}
		  1 & d_i\rightarrow q \in \mathcal E(t)\\
		  0 & d_i\rightarrow q \notin \mathcal E(t)
		  \end{cases}, \quad q\in \{a,\tau\}.
		  \label{eq:h1}
		  \end{equation}}
		  To indicate the visibility of other defenders   {$d_j\in \mathcal D\backslash \{d_i\}$}
		  we define the following  binary matrix function $\mathrm{Ad}:[0,T]\rightarrow \mathbb R^{n\times n}$ 
		  with $ij$th entry defined as
		  \begin{equation}
		  [\mathrm{Ad}(t)]_{ij}=\begin{cases}
		  1 & {d_i}\rightarrow d_j \in \mathcal E(t)\\
		  0 & {d_i}\rightarrow d_j \notin \mathcal E(t). 
		  \end{cases}
		  \label{eq:Adij} 
		  \end{equation}
		  Here, $\mathrm{Ad}(t)$ represents the out-degree adjacency matrix associated with
		  the sub-network of $\mathcal G(t)$ with $\mathcal D$ as the vertex set.	Due to the reduced state space, we can define the $n\times (n+1)$ augmented adjacency matrix as follows

		  		  \begin{subequations}		   
		  	\begin{align}
		  	{\mathcal A} (t)=\begin{bmatrix} \Phi_a(t)+ \mathrm{Ad}(t) &\Phi_\tau(t)\end{bmatrix}, 
		  	\end{align} 
		  	where 
		  	\begin{align}
		  	&\Phi_a(t)=\diag{\phi_{d_1}^a(t),\phi_{d_2}^a(t),\cdots,\phi_{d_n}^a(t)}, \\
		  	&\Phi_\tau(t)=\text{col}\{\phi_{d_1}^\tau(t),\phi_{d_2}^\tau(t),\cdots,\phi_{d_n}^\tau(t)\}.
		  	\end{align}
		  	\label{eq:extadj}
		  \end{subequations}
		  The $i$th row of the matrix $\mathcal A(t)$, denoted by $[\mathcal A(t)]_{i\bullet}$, provides
		  the information about all the players $p\in \mathcal P \backslash \{d_i\}$ who are  visible to $d_i$  at time $t\in[0,T]$. In particular, the $i$th and $(n+1)$th elements of $[\mathcal A(t)]_{i\bullet}$ indicate the visibility of the attacker and the target respectively, and the $j$th element, with $j\neq \{i,n+1\}$ indicates the visibility  to the defender $d_j\in \mathcal D\backslash \{d_i\}$.
		  Using this, we define the implementable or network adapted feedback strategies for defenders as follows.
		  \begin{definition} The set of \textit{network adapted} linear feedback strategies of the defender $d_i$ ($d_i\in \mathcal D$) is given by
		  	\begin{multline}
		  	{\Gamma}^{\text{Ad}}_{d_i}:=\Bigl\{{\gamma}^{\text{Ad}}_{d_i}:[0,T]\times \mathbb{R}^{2(n+1)}\times \mathbb{R}^{2\times 2(n+1)}\rightarrow \mathbb{R}^{2}~\Big| ~{u}_{d_i}(t):=\gamma_{d_i}^{\text{Ad}}(t,z(t);K_{d_i}(t))=K_{d_i}(t) \mathcal I_{d_i}(t)z(t),\\K_{d_i}(t)\in \mathbb{R}^{2\times 2(n+1)},~t\in[0,T]\Bigr\}.
		  	\label{eq:nwadapt}			
		  	\end{multline}
		  	and $\Gamma^\text{Ad}_d=\Gamma^\text{Ad}_{d_1}\times   \cdots \times \Gamma^\text{Ad}_{d_n}$ denotes the set of defender team's joint network adapted strategies.
		  	
		  	Here, the matrix function $\mathcal I_{d_i}:[0,T]\rightarrow \mathbb R^{2(n+1)\times 2(n+1)}$  captures state information of players in $\mathcal P\backslash \{d_i\}$  who are visible to defender $d_i$ at time $t$, which is defined by
		  	\begin{align}
		  	\mathcal I_{d_i}(t):= \diag{[\mathcal A(t)]_{i\bullet}}  \left(\idm{n+1}-{\mathbf{e}^i_{n+1}}^\prime\otimes\left[\odm{n+1}-\mathbf{e}^i_{n+1}\right]\right)\otimes \idm{2},
		  	\label{eq:infomatrix}
		  	\end{align}
		  	and the gain matrix $K_{d_i}(t)$ can be partitioned into $n+1$ block matrices of size $2\times 2$ as follows
		  	\begin{align} K_{d_i}(t)=\begin{bmatrix}K_{d_i}^{d_1}(t) &  \cdots& K_{d_i}^{d_n}(t)& K_{d_i}^\tau(t) \end{bmatrix}.
		  	\label{eq:kpart}
		  	\end{align} 
		  \end{definition} 
		  The defender team's joint network adapted linear feedback strategy  given by
		  \begin{subequations} 
		  	\begin{align} 
		  	&\gamma^{\text{Ad}}_d(t,z(t);K_d(t))=K_d(t) \mathcal I_d(t) z(t),\\
		  	& K_d(t)= \diag{K_{d_1}(t),  \cdots, K_{d_n}(t)}, \label{eq:kdm}\\
		  	& \mathcal I_d(t)= \col{\mathcal I_{d_1}(t), \cdots,\mathcal I_{d_n}(t)}.
		  	\end{align}
		  \end{subequations}  
We illustrate the structure of network adapted feedback strategies with the following example.
\begin{example} Consider the multi-player interaction illustrated in Figure \ref{fig:int_game}. Let $t_1$ be the time instant where the interactions of players result in the visibility network  illustrated by Figure \ref{fig:playerINTgraph}. The matrices \eqref{eq:Adij}, \eqref{eq:extadj}, and \eqref{eq:infomatrix} (for defender $d_2$) are given by
	\begin{align*}
	&\mathrm{Ad}(t_1)= \begin{bmatrix}0 &1 &0 & 0\\1&0&1 & 1\\ 0 & 0 &0 &0\\0 & 0& 1&0  \end{bmatrix},~
	\Phi_a(t_1)=\begin{bmatrix}1&   \\   & 1&   \\   &   & 0 &  \\  &   &   & 0\end{bmatrix},~	 \Phi_\tau(t_1)=\begin{bmatrix}0\\0\\0\\1\end{bmatrix}\\
	& \mathcal A(t_1)=
	\begin{bmatrix}1& 1& 0& 0& 0\\ 1 & 1& 1&1&0\\0&0&0&0&0 \\ 0 & 0& 1&0& 1 \end{bmatrix},~\mathcal I_{d_2}(t_1)=
	\begin{bmatrix}1&\\ & 1&\\&&1&\\&&&1\\&&&&0 \end{bmatrix}   \begin{bmatrix}1&-1&0&0&0\\0& 1&0&0&0\\0&-1&1&0&0\\0&-1&0&1&0\\0&-1&0&0&1
	\end{bmatrix}\otimes \idm{2}. 	
	\end{align*}
	The network adapted feedback strategy for defender  $d_2$, using \eqref{eq:kpart}, is calculated as
	\begin{align*}
	\gamma_{d_2}^{\text{Ad}}(t_1,z(t_1);K_{d_2}(t_1))&=K_{d_2}(t_1)
	\mathcal I_{d_2}(t_1)  z(t_1)\\
	&=K_{d_2}^{d_2}(t_1)z_{d_2}(t_1)+K_{d_2}^{d_1}(t_1)(z_{d_1}(t_1)-z_{d_2}(t_1))\\
	&\quad +K_{d_2}^{d_3}(t_1)(z_{d_3}(t_1)-z_{d_2}(t_1))+K_{d_2}^{d_4}(t_1)(z_{d_4}(t_1)-z_{d_2}(t_1)).
	\end{align*}
	Clearly, defender $d_2$ can implement the above feedback strategy as the players $\{a,d_1,d_3,d_4\}$ are visible at $t_1$.\qed 
\end{example}
%%%%%
 \subsection{ {Network adapted feedback Nash equilibrium  strategies}}
 \label{sec:NAFBNS}
 In this subsection we derive network adapted feedback Nash equilibrium (NAFNE) strategies of the defenders using Theorem \ref{thm:fbne}. 
 \begin{remark}  {Due to feature \textbf{F2}, in interaction \textbf{I1},  the attacker and the target being unaware of  defenders' visibility constraints} becomes common knowledge of the game; see \cite{Myerson:97}. As a result, the outcome of the game is that  the attacker and target will use their FNE strategies $u^*_{p}(t)=\gamma_p^*(t,z(t))=-R_{p}^{-1}B_p^\prime  P_p(t)z(t),~p\in\{a,\tau\}$, for all $t\in[0,T]$ associated with the game where all the players have unlimited observation range. Further, defenders also know that attacker and target will use their FNE strategies. 
 	\label{rem:commonknowledge}
 \end{remark} 
 When the defenders use their   network adapted feedback strategy
 $\gamma^\text{Ad}_d $, it is required that the strategy profile $(\gamma^\text{Ad}_d,\gamma_\tau^*,\gamma_a^*)$ is a FNE. However, this strategy profile cannot be a FNE for  {NZLQDG}, with players' objectives given by \eqref{eq:fullobj}, unless all the defenders have unlimited observation range.    
 To see this, we recall that in a game setting, the objectives of the players are interrelated. Further, we note that the defenders, due to lack of full state information, deviate unilaterally from using \eqref{eq:1d} while implementing   $\gamma_d^\text{Ad}$. From Remark \ref{rem:commonknowledge}, as the attacker and target strategies are fixed at their standard FNE strategies \eqref{eq:1d}, the strategy profile  $(\gamma^\text{Ad}_d,\gamma_\tau^*,\gamma_a^*)$ cannot be a Nash equilibrium.   
 This implies,  the performance indices or objectives of the players for which the strategy profile $(\gamma^\text{Ad}_d,\gamma_\tau^*,\gamma_a^*)$ is a  FNE differs from the objectives  \eqref{eq:fullobj}. 
 
  {In the following theorem, we use inverse game theory approach based on strategies obtained in Theorem \ref{thm:fbne}; see also \cite{lin2015nash} in the context of a PE game}.  In particular, we construct a class of performance indices, parameterized by the gain matrices $K_d(t),~t\in[0,T]$ with respect to which the strategy profile $(\gamma^\text{Ad}_d,\gamma_\tau^*,\gamma_a^*)$  is a NAFNE.
 	%%%
 	%%%
 	\begin{theorem}
 		Consider the $(n+2)$--player finite horizon  {NZLQDG} described by \eqref{eq:state_eq}-\eqref{eq:fullobj}. For an arbitrary gain matrix $K_d(t)=\diag{K_{d_1}(t),K_{d_2}(t),\cdots,K_{d_n}(t)},t\in[0,T]$, the strategy profile  $(\gamma^\text{Ad}_d,\gamma_\tau^*,\gamma_a^*)$,  with  {$\gamma_\tau^*(t,z(t))=-R_\tau^{-1}B_\tau^\prime P_\tau(t)z(t)$}, $\gamma_a^*(t,z(t))=-R_a^{-1}B_a^\prime P_a(t)z(t)$ and $\gamma^\text{Ad}_d(t,z(t);K_d(t))=K_d(t)\mathcal I_d(t)z(t)$, $t\in[0,T]$, forms a NAFNE characterized by the inequalities	
 		\begin{align*}
 		&J^\text{Ad}_d(\gamma^\text{Ad}_d,\gamma_\tau^*,\gamma_a^*)\leq J^\text{Ad}_d(\bar{\gamma}^\text{Ad}_d,\gamma_\tau^*,\gamma_a^*),
 		~\forall \bar{\gamma}^\text{Ad}_d\in  \Gamma^\text{Ad}_d\\
 		&J^\text{Ad}_\tau(\gamma^\text{Ad}_d,\gamma_\tau^*,\gamma_a^*)\leq J^\text{Ad}_\tau(\gamma^\text{Ad}_d,\gamma_\tau,\gamma_a^*),
 		~\forall \gamma_\tau \in \Gamma_\tau\\
 		&J^\text{Ad}_a(\gamma^\text{Ad}_d,\gamma_\tau^*,\gamma_a^*)\leq J^\text{Ad}_a(\gamma^\text{Ad}_d,\gamma_\tau^*,\gamma_a),
 		~\forall \gamma_a \in \Gamma_a,
 		\end{align*}	
 		with parametric performance indices $(J^\text{Ad}_d,J^\text{Ad}_\tau,J^\text{Ad}_a;K_d(.))$ given by
 		\begin{subequations} 	
 			\begin{align}
 			J^\text{Ad}_d(u_d(.),u_\tau(.),u_a(.))&=  \frac{1}{2}||z(T)||^2_{ {F}_{d}}+\frac{1}{2}\int_{0}^{T}\Bigl( ||z(t)||^2_{Q^\text{Ad}_d(t)} \notag \\
 			&\hspace{0.1in}+||u_d(t)||^2_{R_d}-u_d^\prime(t) S_1(t)z(t)  -z^\prime(t) S_1^\prime(t) u_d(t)\Bigr) dt,		\label{eq:Jbs} 
 			\\
 			J^\text{Ad}_\tau(u_d(.),u_\tau(.),u_a(.))&=\frac{1}{2}||z(T)||^2_{ {F}_{\tau}}  +\frac{1}{2}\int_{0}^{T}
 			\left(||z(t)||^2_{Q^\text{Ad}_\tau(t)}+||u_\tau(t)||^2_{R_\tau}\right)dt,
 			\label{eq:Jas} \\
 			J^\text{Ad}_a(u_d(.),u_\tau(.),u_a(.))&=
 			\frac{1}{2}||z(T)||^2_{ {F}_{a}} +\frac{1}{2}\int_{0}^{T}\hspace{-0.09in}\left(||z(t)||^2_{Q^\text{Ad}_a(t)}+||u_a(t)||^2_{R_a}\right)dt,
 			\label{eq:Jcs} 
 			\end{align}
 			\label{eq:Jseq}
 		\end{subequations} 	   
 		where 
 		 { 
 		\begin{subequations} 
 			\begin{align}
 			S_1(t)&= {B}_d^\prime P_d(t)+R_d K_d(t)\mathcal I_d(t)\\
 			\Delta Q^\text{Ad}_\tau(t)&=-P_\tau(t) {B}_dR_d^{-1}S_1(t)-S_1^\prime(t) R_d^{-1} {B}_d^\prime P_\tau(t) \\
 			\Delta Q^\text{Ad}_d(t)&=-P_d(t) {B}_dR_d^{-1} {B}_d^\prime P_d(t)  +\mathcal I_d^\prime(t)K^\prime_d(t) R_dK_d(t)\mathcal I_d(t) \\
 			\Delta Q^\text{Ad}_a(t)&= -P_a(t) {B}_dR_d^{-1}S_1(t)-S_1^\prime(t) R_d^{-1} {B}_d^\prime P_a(t).
 			\end{align} 
 			\label{eq:perfindex}
 		\end{subequations}
 		Here, $\Delta Q_p^\text{Ad}(t)=Q_p^\text{Ad}(t)-Q_p$, $p\in \{d,a,\tau\}$, and
 		 $P_d(t)$, $P_\tau(t)$, and $P_a(t)$ are solutions of symmetric coupled RDE \eqref{eq:Riccati}}.
 		\label{thm:thm1}
 	\end{theorem} 
 	%%%% 
 	\begin{proof}
 		
 		We define the value functions $V_i(t,z(t))$ for $i\in \{d,\tau,a\}$ as
 		\begin{align}
 		V_i(t,z(t))=\frac{1}{2} z^\prime(t)P_i(t)z(t).
 		\end{align}
 		Taking the time derivative of the value function associated with the 
 		cooperative defenders we get
 		\begin{align}\label{eq:vb}
 		\dot{V}_d(t,z(t))=&\frac{1}{2}\dot{z}^\prime(t) P_d(t)z(t)+\frac{1}{2}z^\prime(t) P_d(t)\dot{z}(t)\nonumber\\
 		&+\frac{1}{2}z^\prime(t) \dot{P}_d(t)z(t).
 		\end{align}
 		Substituting for state dynamics in \eqref{eq:vb} we get
 		\begin{multline*}
 		\dot{V}_d(t,z(t))=\frac{1}{2}\left[ {B}_du_d (t) + {B}_\tau u_\tau  (t)   + B_au_a(t) \right]^\prime P_d(t) z(t)\\+\frac{1}{2}z^\prime(t) P_d(t)\left[{B}_du_d(t)+{B}_\tau u_\tau(t)+{B}_au_a(t)\right] \\
 		+\frac{1}{2}z^\prime(t)\bigl[-{Q}_d+P_d (t)S_d P_d(t)
 		+P_d(t)S_\tau P_\tau(t)\\+P_d(t) S_a P_a(t)+P_\tau (t)S_\tau P_d(t)+P_a(t)S_a P_d(t)\bigr]z(t).
 		\end{multline*}
 		Using  \eqref{eq:perfindex}, terms in the above equation can be rearranged as 
 		\begin{multline*}
 		\dot{V}_d(t,z(t))=\frac{1}{2}||u_d(t)-K_d(t)\mathcal I_d(t)z(t)||^2_{R_d}-\frac{1}{2}||u_d(t)||^2_{R_d}\\
 		+z^\prime(t) P_d(t) {B}_\tau \left[u_\tau(t)+R_\tau^{-1} {B}_\tau^\prime P_\tau(t)z(t)\right]\\  	   
 		+z^\prime(t) P_d(t) {B}_a\left[u_a(t)+R_a^{-1} {B}_a^\prime P_a(t)z(t)\right]\\
 		-\frac{1}{2} ||z(t)||^2_{Q^\text{Ad}_d(t)}
 		+\frac{1}{2}u_d^\prime(t)S_1(t)z(t)+\frac{1}{2}z^\prime(t)S_1^\prime(t)u_d(t).
 		\end{multline*}
 		Integrating the above equation from $0$ to $T$ and rearranging terms we get
 		\begin{multline*}
 		V_d(T,z(T))+\frac{1}{2}\int_0^T \Bigl(||u_d(t)||^2_{R_d} + ||z(t)||^2_{Q^\text{Ad}_d(t)} 
 		-u_d^\prime(t)S_1(t)z(t)-z^\prime(t)S_1^\prime(t)u_d(t)\Bigr)dt\\
 		=V_d(0,z(0))
 		+ \frac{1}{2}\int_0^T\Bigl(
 		||u_d(t)-K_d(t)\mathcal I_d(t)z(t)||^2_{R_d}\\
 		+2z^\prime(t) P_d(t) {B}_\tau\left[u_\tau(t)
 		+R_\tau^{-1} {B}_\tau^\prime P_\tau(t)z(t)\right]\\  	   
 		+2z^\prime(t) P_d(t) {B}_a\left[u_a(t)+R_a^{-1} {B}_a^\prime P_a(t)z(t)\right]\Bigr) dt.
 		\end{multline*}
 		As $V_d(T,z(T))=\frac{1}{2} z^\prime(T)P_d(T)z(T)$ and $P_d(T)=F_d$, thus we get:
 		\begin{multline}
 		J^\text{Ad}_d(u_d(.),u_\tau(.),u_a(.))=V_d(0,z(0))  
 		+\frac{1}{2}\int_0^T\Bigl(
 		2z^\prime(t) P_d(t){B}_\tau\left[u_\tau(t)-\gamma_\tau^*(t,z(t))\right] \\  	   
 		+||u_d(t)-\gamma_d^\text{Ad}(t,z(t);K_d(t))||^2_{R_d}\\+2z^\prime(t) P_d(t) {B}_a\left[u_a(t)
 		-\gamma_a^*(t,z(t))\right]\Bigr) dt,  
 		\label{eq:Jbseq}
 		\end{multline}
 		where $J^\text{Ad}_d(u_d(.),u_\tau(.),u_a(.))$ is defined in \eqref{eq:Jbs}. Using the same approach as above we can show the following relations
 		\begin{multline}
 		J^\text{Ad}_\tau(u_d(.),u_\tau(.),u_a(.))=V_t(0,z(0)) +\frac{1}{2}\int_0^T\Bigl(
 		||u_\tau(t)-\gamma_\tau^*(t,z(t))||^2_{R_\tau}\\
 		+2z^\prime(t) P_\tau(t) {B}_d\left[u_d(t)-\gamma_d^\text{Ad}(t,z(t);K_d(t))\right]\\+
 		2z^\prime(t) P_\tau(t) {B}_a\left[u_a(t)-\gamma_a^*(t,z(t))\right]
 		\Bigr) dt,  
 		\label{eq:Jaseq}
 		\end{multline}
 		\begin{multline}
 		J^\text{Ad}_a(u_d(.),u_\tau(.),u_a(.))=V_a(0,z(0)) +\frac{1}{2}\int_0^T\Bigl(
 		2z^\prime(t) P_a(t) {B}_\tau\left[u_\tau(t)-\gamma_\tau^*(t,z(t))\right]\\
 		+2z^\prime(t) P_a(t) {B}_d\left[u_d(t)-\gamma_d^\text{Ad}(t,z(t);K_d(t))\right]+\\
 		||u_a(t)-\gamma_a^*(t,z(t))||^2_{R_a}\Bigr) dt.  
 		\label{eq:Jcseq}
 		\end{multline}
 		Clearly,  $(\gamma^\text{Ad}_d,\gamma_\tau^*,\gamma_a^*)$ is a NAFNE of the game with performance indices \eqref{eq:Jseq}.
 	\end{proof}
 	%%%
 	
 	%%%% 
 	\begin{remark} 	\label{rem:Bayesian} {Recalling feature \textbf{F2} we note that the interaction \textbf{I1} is a game of asymmetric information.} In the language of Bayesian games \cite{Myerson:97},  this implies that the attacker's and target's beliefs, over defenders' type set, would assign probability equal to one to the type where defenders' have unlimited observations. As a result, the attacker and the target use their standard FNE strategies associated with the performance indices $(J_d,J_\tau,J_a)$, whereas the defenders use their NAFNE strategies associated with the parametric performance indices $(J^\text{Ad}_d,J^\text{Ad}_\tau,J^\text{Ad}_a;K_d(.))$. 	
 	\end{remark}
 	%%%
  
 Theorem \ref{thm:thm1} characterizes  parametric performance indices $(J^\text{Ad}_d,J^\text{Ad}_\tau,J^\text{Ad}_a;K_d(.))$   for which the strategy profile $(\gamma^\text{Ad}_d,\gamma_\tau^*,\gamma_a^*)$ is a NAFNE.  Since the choice of gain matrices $K_d(t),~t\in[0,T]$ is arbitrary, we obtain a very large class of performance indices. In section \ref{sec:optimization}, we develop a consistency criterion for selecting the gain matrices.  
 
 When the attacker is non-suicidal, the interactions inherently involve two simultaneous PE games. First one involving the attacker and the defenders' team, and the second one involving the attacker and the target. Though we obtain the implementable strategies, through  Theorem \ref{thm:thm1}, it is difficult to geometrically characterize the trajectories of the players. However, 	when the attacker is suicidal,  then the defenders are only reacting to a single PE interaction involving the  attacker and the target. In the next theorem, we recover the classical result \cite{isaacs65a} that the trajectories of the attacker and the target evolve along a straight line.  
 	\begin{theorem}\label{thm:straight_line} Consider the $(n+2)$--player finite horizon {NZLQDG} described by \eqref{eq:state_eq}-\eqref{eq:fullobj} with $\lambda=0$. Then, the suicidal attacker and the target move on the straight line joining their locations at time $t=0$. Further, the visibility constraints of the defenders have no effect
 		on the control and state trajectories of the attacker and the target.
 	\end{theorem}   
    \begin{proof}
	    	As the attacker is suicidal, it is sufficient to consider interactions with one defender, labeled by $d_1$.
	    	The state equation \eqref{eq:state_eq} using the network adapted feedback Nash equilibrium $(\gamma_{d_1}^\text{Ad},\gamma_\tau^*,\gamma_a^*)$ is given by 
	    	\begin{align}\label{eq:state_eq_rewrite}
	    	\dot{z}(t)=\left(B_{d_1}K_{d_1}(t)\mathcal I_{d_1}(t)-S_\tau P_\tau(t)-S_a P_a(t)\right)z(t),
	    	\end{align} 
	    	Following symmetry property, we partition the matrix $P_i(t)$ as $P_i(t)=\begin{bmatrix}
	    	P_i^{11}(t)&P_i^{12}(t)\\
	    	P_i^{12}(t)&P_i^{22}(t)
	    	\end{bmatrix}$ for $i \in \{{d_1},\tau,a\}$. Using this we obtain
	    	\begin{align}\label{eq:z_t_first}
	    	\dot{z}_\tau(t)=& - \Bigl[r_{\tau}^{-1}P_\tau^{12}(t)+r_a^{-1}(P_a^{11}(t)+P_a^{12}(t))\Bigr]z_{d_1}(t)- \Bigl[r_{\tau}^{-1}P_\tau^{22}(t)+r_a^{-1}(P_a^{12}(t)+P_a^{22}(t))\Bigr]z_\tau(t).
	    	\end{align}
	    	The   RDE \eqref{eq:Riccati} can be partitioned and collected as
	    	\begin{multline}
	    	\begin{bmatrix}
	    	\dot{P}_\tau^{11}(t) & \dot{P}_\tau^{12}(t) & \dot{P}_a^{11}(t) & \dot{P}_a^{12}(t) 
	    	\end{bmatrix} 
	    	= \begin{bmatrix}
	    	{P}_\tau^{11}(t) & {P}_\tau^{12}(t) & {P}_a^{11}(t) & {P}_a^{12}(t) 
	    	\end{bmatrix}
	    	\Sigma(t)\\
	    	+ r_{d_1}^{-1}P_{d_1}^{11}(t)\begin{bmatrix}
	    	{P}_\tau^{11}(t) & {P}_\tau^{12}(t) & {P}_a^{11}(t) & {P}_a^{12}(t) 
	    	\end{bmatrix},
	    	\end{multline}
	    	where the matrix $\Sigma(t)$ is partitioned as $4\times 4$ block matrix with elements 
	    	$\Sigma^{11}(t)=\Sigma^{33}(t)=
	    	r_{d_1}^{-1}P_{d_1}^{11}(t)+r_a^{-1}(P_a^{11}(t)+ P_a^{12}(t))$,
	    	$\Sigma^{12}(t)=\Sigma^{34}(t)=r_{d_1}^{-1}P_{d_1}^{12}(t)+r_a^{-1}(P_a^{12}(t)+ P_a^{22}(t))$,
	    	$\Sigma^{21}(t)=\Sigma^{43}(t)=r_{\tau}^{-1}P_\tau^{12}(t)+r_a^{-1}(P_a^{11}(t)+P_a^{12}(t))$,
	    	$\Sigma^{22}(t)=\Sigma^{44}(t)=r_{\tau}^{-1}P_\tau^{22}(t)+r_a^{-1}(P_a^{12}(t)+P_a^{22}(t))$,
	    	$\Sigma^{31}(t)=\Sigma^{41}(t)=r_a^{-1}(P_\tau^{11}+P_\tau^{12}(t))$,
	    	$\Sigma^{13}(t)=\Sigma^{14}(t)=0$, and
	    	$\Sigma^{32}(t)=\Sigma^{42}(t)=r_a^{-1}(P_\tau^{12}+P_\tau^{22}(t))$.
	    	Since ${P}_\tau^{11}(T)={P}_\tau^{12}(T)={P}_a^{11}(T)={P}_a^{12}(T)=0$, then 
	    	following  the matrix variation of constant formula 
	    	\cite[Theorem 1, pg. 59]{Brockett:15},  it follows immediately that ${P}_\tau^{11}(t)= {P}_\tau^{12}(t)={P}_a^{11}(t)={P}_a^{12}(t)=0$ for all $t\in[0,T]$. Using this in \eqref{eq:z_t_first} we get
	    	\begin{align}\label{eq:z_t_second}
	    	\dot{z}_\tau(t)=& -\bigl[r_{\tau}^{-1}P_\tau^{22}(t)+r_a^{-1}P_a^{22}(t)\bigr]z_\tau(t).
	    	\end{align}
	    	Again, using symmetry property we partition the matrices $P_\tau^{22}(t)$ and $P_a^{22}(t)$ as $P_\tau^{22}(t)=\begin{bmatrix}
	    	k_1(t) & k_2(t)\\
	    	k_2(t) & k_3(t)
	    	\end{bmatrix}$ and $P_a^{22}(t)=\begin{bmatrix}
	    	k_4(t) & k_5(t)\\
	    	k_5(t) & k_6(t)
	    	\end{bmatrix}$, we obtain
	    	\begin{align*}
	    	\dot{k}_1(t)&=q_\tau 
	    	+r_{\tau}^{-1}(k_1^2(t)+k_2^2(t))+2r_a^{-1}(k_1(t)k_4(t)+k_2(t)k_5(t))\\
	    	\dot{k}_2(t)&=r_{\tau}^{-1}k_2(k_1+k_3)+r_a^{-1}(k_2(k_4+k_6)+k_5(k_1+k_3))\\
	    	\dot{k}_3(t)&=q_\tau 
	    	+r_{\tau}^{-1}(k_2^2(t)+k_3^2(t))+2r_a^{-1}(k_2(t)k_5(t)+k_3(t)k_6(t))\\		
	    	\dot{k}_4(t)&=-q_{a\tau}
	    	+r_{a}^{-1}(k_4^2(t)+k_5^2(t))+2r_\tau^{-1}(k_1(t)k_4(t)+k_2(t)k_5(t))\\
	    	\dot{k}_5(t)&=r_{a}^{-1}k_5(t)(k_4(t)+k_6(t))+r_\tau^{-1}(k_2(t)(k_4(t)+k_6(t))+k_5(t)(k_1(t)+k_3(t)))\\
	    	\dot{k}_6(t)&=-q_{a\tau}
	    	+r_{a}^{-1}(k_5^2(t)+k_6^2(t)) +2r_\tau^{-1}(k_2(t)k_5(t)+k_3(t)k_6(t)),
	    	\end{align*}
	    	where $k_1(T)=k_3(T)=-f_{\tau a}$, $k_2(T)=k_5(T)=0$, $k_4(T)=k_6(T)=f_{a\tau}$. Since $k_2(T)=k_5(T)=0$, $k_1(T)-k_3(T)=0$, and $k_4(T)-k_6(T)=0$, using again the matrix variation of 
	    	constants formula we can show $k_2(t)=k_5(t)=0$, $k_1(t)=k_3(t)$, and $k_4(t)=k_6(t)$ for all $t\in[0, T]$. 
	    	
	    	So, $P_\tau^{22}=k_1(t)\idm{2}$ and $P_a^{22}=k_4(t)\idm{2}$. Using this in \eqref{eq:z_t_second} we get
	    	\begin{align}\label{eq:z_t_third}
	    	\dot{z}_\tau(t)= - \left(r_\tau^{-1}k_1(t)+r_a^{-1}k_4(t)\right) z_\tau(t).
	    	\end{align} 
	    	Representing $z_\tau=[z_\tau^x,~z_\tau^y]^\prime$ the slope of the line joining the attacker $a$ and the target $\tau$  at time $t$  is given by
	    	$s(t)=\frac{z_\tau^y(t)}{z_\tau^x(t)}$ for $z_\tau^x(t)\neq 0$.
	    	The time derivative of the slope $s(t)$ results in
	    	\begin{align*}
	    	\dot{s}(t)&=\frac{\dot{z}_\tau^y(t) z_\tau^x(t)- z_\tau^y(t)\dot{z}_\tau^x(t)}{(z_\tau^x(t))^2}\\&= - \left(r_t^{-1}k_1(t)+r_a^{-1}k_4(t)\right)\frac{ z_\tau^x(t) z_\tau^y(t)-  z_\tau^x(t) z_\tau^y(t)}{(z_\tau^x(t))^2}\\&   =0.
	    	\end{align*} 
	    	Finally, when $z_\tau^x(t)=0$ then $z_\tau^x(s)=0$ for all $s \in [t,T]$ this implies that attacker and target move along the $y$-axis 
	    	during the time period $[t,T]$. 
	    	
	    	Since the attacker and the target use their standard FNE strategies, and from the solution of $P_a(t)$ and $P_\tau(t)$ obtained from the above, we get
	    	\begin{subequations} 
	    		\begin{align}
	    		u_a^*(t)&=-R_a^{-1}B_a^\prime P_a(t)z(t)=r_a^{-1}k_4(t)z_\tau(t),\\
	    		u_\tau^*(t)&=-R_\tau^{-1}B_\tau^\prime P_\tau(t)z(t)=-r_\tau^{-1}k_1(t)z_\tau(t).
	    		\end{align}
	    		\label{eq:controlat}
	    	\end{subequations} 
	    	Clearly, from \eqref{eq:z_t_third} and \eqref{eq:controlat}, the visibility constraints of the defenders have no effect on the control and state trajectories of the attacker and the target.
	    \end{proof}
  	\begin{remark}  In Theorem \ref{thm:straight_line}, though the defenders' visibility constraints have no effect on strategies of the attacker and target, they can  influence the eventual outcome and the termination time of the game. 
  \end{remark}
	    %%%%
	    %%
	    %%		 
	    %%	    
	    \section{ {Analysis of interaction \textbf{I2}}}\label{sec:feedback_zsg}
	     {In this section, we analyze interactions \textbf{I2} where the visibly constrained defenders and target cooperate as a team against the non-suicidal attacker. Following a similar approach  developed in section \ref{sec:feedback_nzsg}, we model this interaction as a two-player zero-sum linear quadratic differential game (ZLQDG) \cite{Basar:99} with attacker as the first player, and the team of defenders and the target as the second player}. The dynamics \eqref{eq:state_eq} of the players can be rewritten as 
	    \begin{align}\label{eq:state_eq1}
	    \dot{z}(t)=B_{d\tau}u_{d\tau}(t)
	    +B_au_a(t),
	    \end{align}
	    where $u_{d\tau}(t)=\col{u_d(t),u_\tau(t)}$,~$t\in[0,T]$, $B_{d\tau}=[B_d~ B_\tau]$. 	
	    The objective function minimized by the attacker  and maximized by the team of defenders and  the target is given by
	    \begin{align}
	    J(u_{d\tau}(.),u_a(.))&=  \frac{1}{2}||z_{\tau}(T)||^2_{F_{a\tau}} -\frac{1}{2}\sum_{j=1}^{n}||z_{d_j}(T)||^2_{F_{ad_j}}\nonumber\\
	    &
	    +\frac{1}{2}\int_{0}^{T} \Big( ||z_{\tau}(t)||^2_{Q_{a\tau}}   -\sum_{j=1}^{n}||z_{d_j}(t)||^2_{Q_{ad_j}}
	    + ||u_a(t)||^2_{R_a}
	    \nonumber\\
	    &\quad \quad -  ||u_\tau(t)||^2_{R_\tau} -\sum_{j=1}^{n} ||u_{d_j}(t)||^2_{R_{d_j}} 
	    \Big)   dt, \nonumber\\
	    =&\frac{1}{2}||z(T)||_{F}^2+\frac{1}{2}\int_{0}^{T} \Big( ||u_a(t)||^2_{R_a} 
	      - ||u_{d\tau} (t)||^2_{ R_{d\tau}} + ||z(t)||^2_{Q}\Big)   dt,  \label{eq:obj:zero_sum}
	    \end{align} 
	    	where  $R_{d\tau}=\diag{R_{d_1}, ,\cdots ,R_{d_n},R_\tau}$,
	    $F=\text{diag}\{-F_{ad_1 },\cdots,-F_{ad_n}, F_{a\tau}\}$ and \\$Q=\diag{-Q_{ad_1 },\cdots,-Q_{ad_n},Q_{a\tau}}$. 	    
	    Let $\gamma_{d\tau}=\col{\gamma_d,\gamma_\tau} \in\Gamma_d \times \Gamma_\tau$ represents the feedback strategy set of the defenders and target team. The strategy profile $(\gamma^*_{d\tau},  \gamma^*_{a})$ is a FNE for the  {ZLQDG} if the following set of inequalities hold true
	    \begin{subequations} 
	    	\begin{align}
	    	&J(\gamma^*_{d\tau},\gamma^*_a) \geq J(\gamma_{d\tau}, \gamma^*_a), ~\forall \gamma_{d\tau}\in  \Gamma_{d\tau},\\
	    	&J(\gamma^*_{d\tau},\gamma^*_a) \leq J(\gamma^*_{d\tau}, \gamma_a), ~\forall \gamma_a\in  \Gamma_a.
	    	\end{align}
	    \end{subequations} 
	    The next theorem from \cite{engwerda2005lq} provides conditions for the existence of FNE
	    associated with  {ZLQDG}.
	    \begin{theorem}{\cite[Theorem 8.4]{engwerda2005lq}}\label{thm:zero:sum:game}
	    	Consider the  {ZLQDG} described by \eqref{eq:state_eq1}-\eqref{eq:obj:zero_sum}. This 
	    	game has a FNE, denoted by $(\gamma^*_{d\tau},\gamma^*_a)$, for every initial state, if and only if the following RDE has a symmetric solution $P(t)$ on $[0,T]$
	    	\begin{align}\label{eq:RDE}
	    	\dot{P}(t)=-Q+P(t)\left(S_a-S_{d\tau}\right)P(t),~P(T)=F.
	    	\end{align}
	    	{where $S_i=B_iR_i^{-1}B_i^\prime$,~$i=\{d\tau,a\}$.}		Moreover, if equation \eqref{eq:RDE} has a solution, the game has a unique equilibrium. The
	    	equilibrium actions are given by
	    	\begin{subequations}
	    		\begin{align}
	    		u^*_a(t)&=\gamma_a^*(t,z(t))=-R_a^{-1}B_a^\prime P(t)z(t), \label{eq:zero_sum_control1}\\
	    		u^*_{d\tau}(t)&=\gamma^*_{d\tau}(t,z(t))=R_{d\tau}^{-1}B_{d\tau}^\prime P(t)z(t).\label{eq:zero_sum_control2}
	    		\end{align}
	    	\end{subequations}
	    \end{theorem}
	    The conditions under which the RDE \eqref{eq:RDE} admits a solution follow from
	    \cite[Corollary 5.13]{engwerda2005lq}. The equilibrium team strategy \eqref{eq:zero_sum_control2} can be decomposed as follows
	     {\begin{subequations} 
	    	\begin{align}
	    	u_{d_i}^*(t)&=r_{d_i}^{-1}\Big[\sum_{j=1}^{n+1}P^{ij}(t) \Big] z_{d_i}(t)  +r_{d_i}^{-1}P^{i(n+1)}(t)\big(z_{\tau}(t)-z_{d_i}(t)\big)\nonumber \\
	    	&\quad +r_{d_i}^{-1}
	    	\sum_{j=1}^{n}P^{ij}(t) \big(z_{d_j}(t)-z_{d_i}(t)\big),\label{eq:control_def_zsg1}\\
	    	u_{\tau}^*(t)&=r_\tau^{-1}\sum_{j=1}^{n} P^{(n+1)j}(t)\big(z_{d_j}(t)-z_\tau(t)\big)  + r_\tau^{-1}\sum_{j=1}^{n+1} P^{(n+1)j}(t)z_{\tau}(t),
	    	\label{eq:control_target_zsg}
	    	\end{align}
	    	\label{eq:zerosum_team}
	    \end{subequations}} 
	    where matrix $P^{ij}(t)$ is the $i$th row and $j$th column element (a block matrix)
	    obtained by partitioning the matrix $P(t)$ into block matrices of dimension $2\times 2$.
	    
	     {The attacker can implement their FNE strategies \eqref{eq:zero_sum_control1} due to unlimited observation range, whereas the defenders and target cannot implement the FNE strategy \eqref{eq:control_def_zsg1} and  \eqref{eq:control_target_zsg} due to visibility constraints. %
	    	For the  visibility constrained target, to indicate the visibility of the attacker or defender $d_i\in \mathcal D$  at time $t\in[0,T]$, we define the binary function $\phi^p_{\tau}:[0,T]\rightarrow \mathbb R$, $p\in\mathcal D \cup \{a\}$  as}
	    	 {\begin{align}
	    	\phi^p_{\tau}(t)=\begin{cases}
	    	1 & \tau \rightarrow p \in \mathcal E(t)\\
	    	0 & \tau \rightarrow p \notin \mathcal E(t). 
	    	\end{cases} 
	    	\label{eq:tauh1}
	    	\end{align} 
	       	The set of \textit{network adapted} linear feedback strategies of the target $\tau$ is given by
	    	\begin{multline}
	    	{\Gamma}^{\text{Ad}}_{\tau}:=\Bigl\{{\gamma}^{\text{Ad}}_{\tau}:[0,T]\times \mathbb{R}^{2(n+1)}\times \mathbb{R}^{2\times 2(n+1)}\rightarrow \mathbb{R}^{2}~\Big|\\~{u}_{\tau}(t):=\gamma_{\tau}^{\text{Ad}}(t,z(t);K_{\tau}(t))=K_{\tau}(t) \mathcal I_{\tau}(t)z(t),\\K_{\tau}(t)\in \mathbb{R}^{2\times 2(n+1)},~t\in[0,T]\Bigr\}.
	    	\label{eq:nwadapttau}			
	    	\end{multline}
	    	Here, the matrix function $\mathcal I_{\tau}:[0,T]\rightarrow \mathbb R^{2(n+1)\times 2(n+1)}$  captures state information of players in $\mathcal P\backslash \{\tau\}$  who are visible to target at time $t$, which is defined by
	    	\begin{align*}
	    	\mathcal I_{\tau}(t):= \Phi_\tau(t) \left(\idm{n+1}-{\mathbf{e}^{n+1}_{n+1}}^\prime\otimes\left[\odm{n+1}-\mathbf{e}^{n+1}_{n+1}\right]\right)\otimes \idm{2},
	    	\label{eq:infomatrixtau}
	    	\end{align*}
	    	where	$\Phi_\tau(t)=\diag{\phi_\tau^{d_1}(t), \cdots,\phi_\tau^{d_n}(t),\phi_\tau^{a}(t)}$
	    	and the gain matrix $K_{\tau}(t)$ can be partitioned into $n+1$ block matrices of size $2\times 2$ as follows
	    	\begin{align} K_{\tau}(t)=\begin{bmatrix}K_{\tau}^{d_1}(t) &  \cdots& K_{\tau}^{d_n}(t)& K_{\tau}^a(t) \end{bmatrix}.
	    	\end{align}} 
	    	 {The network adapted feedback team strategy set is denoted by $\Gamma^\text{Ad}_{d\tau}:=\Gamma^\text{Ad}_d\times \Gamma^\text{Ad}_\tau$. Following  Remark \ref{rem:commonknowledge}, we require that an arbitrary network adapted feedback strategy of the defenders $\gamma_d^\text{Ad}(t,z(t))=K_d(t) \mathcal I_d(t)z(t)$, when coupled with the target's network adapted feedback strategy strategy \eqref{eq:nwadapttau}  as
	    \begin{align}
	    \gamma^\text{Ad}_{d\tau}(t,z(t);(K_{d}(t),K_{\tau}(t)))= \begin{bmatrix}\gamma^\text{Ad}_d(t,z(t);K_d(t))\\
	    \gamma_\tau^\text{Ad}(t,z(t);K_\tau(t))  \end{bmatrix},\label{eq:combined_control_zsg}
	    \end{align}
	    forms a FNE along with the attacker's strategy \eqref{eq:zero_sum_control1}.
	    However, the strategy profile $(\gamma^\text{Ad}_{d\tau},\gamma_a^*)$ cannot be a FNE  for the performance indices \eqref{eq:obj:zero_sum} due to defenders' visibility constraints. Similar to Theorem \ref{thm:thm1}, in the next theorem we construct a class of performance indices, parameterized by the gain matrices $K_d(t),~t\in[0,T]$ with respect to which the strategy profile $(\gamma^\text{Ad}_{d\tau},\gamma_a^*)$  is a NAFNE.}
	  {\begin{theorem}\label{thm:thm2}  
 	Consider 2-player ZLQDG described by \eqref{eq:state_eq1}-\eqref{eq:obj:zero_sum}. For an arbitrary gain matrices $K_d(t)=\diag{K_{d_1}(t), \cdots,K_{d_n}(t)}$,   and $K_\tau(t)$, $t\in[0,T]$ the strategy profile
	    	$(\gamma^\text{Ad}_{d\tau},\gamma_a^*)$   with   $\gamma_a^*(t,z(t))=-R_a^{-1}B_a^\prime P(t)z(t)$ and $\gamma^\text{Ad}_{d\tau}(t,z(t);(K_d(t),K_\tau(t)))$   
	    	forms a NAFNE characterized by the inequalities
	    	\begin{subequations} 
	    		\begin{align} &J^\text{Ad}({\gamma}^\text{Ad}_{d\tau}, \gamma^*_a)\geq J^\text{Ad}(\bar{\gamma}^\text{Ad}_{d\tau}, \gamma^*_a), ~\forall  \bar{\gamma}^\text{Ad}_{d\tau}\in {\Gamma}^\text{Ad}_{d\tau}\\
	    		&J^\text{Ad}({\gamma}^\text{Ad}_{d\tau}, \gamma^*_a)\leq J^\text{Ad}({\gamma}^\text{Ad}_{d\tau}, \gamma_a), ~\forall \gamma_a\in \Gamma_a, 
	    		\end{align} 
	    		\label{eq:zerosumgame} 
	    	\end{subequations} 
	    	with the parametric performance index $(J^\text{Ad};(K_d(.),K_\tau(.)))$ given by 
	    	\begin{multline}\label{eq:obj:zero_sum_stru}
	    	J^\text{Ad}(u_{d\tau}(.),u_a(.))=\frac{1}{2}||z(T)||_{F}^2 
	    	+\frac{1}{2} \int_{0}^{T}\Big(||z(t)||^2_{Q^\text{Ad}(t)}\\
	    	+u_d^\prime(t) S_2(t) z(t) + z^\prime(t) S_2^\prime(t) u_d(t) 
	    	+u_\tau^\prime(t) S_3(t) z(t)\\ + z^\prime(t) S_3^\prime(t) u_\tau(t)+||u_a(t)||^2_{R_a} -||u_{d\tau}(t)||^2_{R_{d\tau}}
	    	\Big)dt,
	    	\end{multline}
	    	where
	    	\begin{subequations} 
	    		\begin{align}
	    		\Delta Q^\text{Ad}(t)&=P(t)B_{d}R_{d}^{-1}B_{d}^\prime P(t)+ P(t)B_{\tau}R_{\tau}^{-1}B_{\tau}^\prime P(t)\nonumber \\
	    		&	    		\hspace{-0.5in} -\mathcal I^\prime_d(t)K_d^\prime(t)R_d K_d(t)\mathcal I_d(t)\label{Q^s(t)}
 -\mathcal I^\prime_\tau(t)K_\tau^\prime(t)R_\tau K_\tau(t)\mathcal I_\tau(t)\\
	    		S_2(t)&=R_dK_d(t)\mathcal I_d(t)- B^\prime_d  P(t)\label{S(t)},\\
	    		S_3(t)&=R_\tau K_\tau(t)\mathcal I_\tau(t)- B^\prime_\tau  P(t)\label{S3(t)}.
	    		\end{align}
	    		\label{eq:perfmatzsg}
	    	\end{subequations} 
    	Here, $\Delta Q^\text{Ad}(t)=Q^\text{Ad}(t)-Q$ and  $P(t)$ is the solution of the RDE \eqref{eq:RDE}.
	    \end{theorem}} 

    \begin{proof} { 
    	Consider the combined value function of defender and target is $V_{d\tau}(t,z(t))$ which is defined as:
    	\begin{align}
    	V_{d\tau}(t,z(t))=\frac{1}{2} z^\prime(t)P(t)z(t).
    	\end{align}
    	Upon differentiating the above equation with respect to $t$, and using \eqref{eq:state_eq1} and \eqref{eq:RDE} we get
    	\begin{align}
    	\dot{V}_{d\tau}(t,z(t))&=\frac{1}{2} \big[ B_{d\tau}u_{d\tau}(t)+ B_a u_a(t)  \big]^\prime P(t)z(t) +\frac{1}{2}z^\prime(t) P(t) \big[ B_{d\tau}u_{d\tau}(t)+ B_a u_a(t) \big]\nonumber\\
    	&\quad +\frac{1}{2}z^\prime(t)\left(-Q+P(t) \left(S_a-S_{d\tau}\right)P(t)\right)z(t),
    	\end{align}
    	Rearranging the above equation, using a few algebraic manipulations, and upon integrating both sides from $0$ to $T$, we get
    	\begin{multline*}
    	V_{d\tau}(T,z(T))=	V_{d\tau}(0,z(0))+\frac{1}{2} \int_{0}^{T}\Big(-
    	||u_\tau- K_\tau(t)\mathcal I_\tau(t) z(t)||_{R_\tau}^2\\
    	+||u_a+R_{a}^{-1}B_a^\prime P(t)z(t)||_{R_a}^2	-
    	||u_d-K_d(t)\mathcal I_d(t)z(t)||_{R_d}^2\\
    	+ \bigl[ -u_d^\prime(t)S_2(t)z(t)  -z^\prime(t)S_2^\prime(t)u_d(t)   -u_\tau^\prime(t)S_3(t)z(t)  -z^\prime(t)S_3^\prime(t)u_\tau(t)   \bigr] +||u_d(t)||^2_{R_d}\\
    	+||u_\tau(t)||^2_{R_\tau}-||u_a(t)||^2_{R_a} 
    	-||z(t)||^2_{Q^\text{Ad}(t)} \Big)dt, 
    	\end{multline*}
    	where $Q^\text{Ad}(t)$, $S_2(t)$ and $S_3(t)$ are as in \eqref{Q^s(t)}, \eqref{S(t)} and  \eqref{S3(t)}. As $V_{d\tau}(T,z(T))=\frac{1}{2}z^\prime (T)P(T)z(T)=\frac{1}{2}||z(T)||^2_F$, the above equation
    	can be written as
    	\begin{multline}\label{eq:j_zsg}
    	J^\text{Ad}(u_{d\tau}(.),u_a(.))=V_{d\tau}(0,z(0))-\frac{1}{2}\int_{0}^{T}\Big(- 
    	||u_a(t)-\gamma_a^*(t,z(t))||_{R_a}^2\\
    	+ 
    	||u_{d\tau}(t)-\gamma^\text{Ad}_{d\tau}(t,z(t);K_d(t),K_\tau(t))||_{R_{d\tau}}^2
    	\Big)dt. 
    	\end{multline}
    	Clearly,  $(\gamma^\text{Ad}_{d\tau},\gamma_a^*)$ is a NAFNE of the zero-sum game with performance index \eqref{eq:obj:zero_sum_stru}.}
    \end{proof}

 \begin{remark}  
	Similar to Remark \ref{rem:Bayesian}, the feature \textbf{F2}  in interaction \textbf{I2} results in a game of  asymmetric information. As a result, the attacker uses its standard FNE strategies associated with the performance index $J$, whereas the defenders and target use their NAFNE strategies associated with the parametric performance index $(J^\text{Ad};(K_d(.),K_\tau(.)))$. 
\end{remark} 

\section{ {Synthesis of network adapted feedback Nash equilibrium strategies}}		\label{sec:optimization}		
 The NAFNE strategies obtained from Theorem \ref{thm:thm1} (Theorem \ref{thm:thm2}) are parameterized by arbitrary gain matrices $K_d(t), (K_d(t), K_\tau(t)),~ t\in[0,T]$ leading to 
a plethora of implementable strategies for the visibility constrained players. To address this issue, we develop an information consistency criterion for selecting  a subset, also referred to as a refinement, of NAFNE strategies. 	The main idea of this  refinement procedure is that, whenever the information is symmetric, that is defenders in interaction \textbf{I1} (defenders and target in interaction \textbf{I2}) are able to see all the players, we require that the defenders' controls (defenders' and target's controls), at those time instants,  using NAFNE strategy must coincide with those using a standard FNE strategy. We formalize this (informational) consistency property in the following definition. 
\begin{definition}{ 
		Let $t_1\in[0,T]$ be a time instant when the defenders in interaction \textbf{I1} (defenders and target in interaction \textbf{I2}) can see all the players in the game process. A NAFNE strategy, parameterized by the gain matrices $K_d(t)$ ($K_d(t), K_\tau(t)$)~$t\in[0,T]$ is \emph{consistent}, and denoted by c-NAFNE, if the control $u_d(t_1)$ ($u_{d\tau}(t_1)$) satisfies  $u_d(t_1)=\gamma_d^\text{Ad}(t_1,z(t_1);K_d(t_1))=\gamma^*_d(t_1,z(t_1))$
		($u_{d\tau}(t_1)=\gamma_{d\tau}^\text{Ad}(t_1,z(t_1);K_d(t_1), K_\tau(t_1))=\gamma^*_{d\tau}(t_1,z(t_1))$).}
	\label{def:consistency}
\end{definition}

In the next theorem, we provide a method for computing the c-NAFNE strategies. First, we introduce the following error function  
{	\begin{align}
		\Theta_1(t)= \gamma_1||\Delta Q^\text{Ad}_d(t)||^2_f+
		\gamma_2||\Delta Q^\text{Ad}_\tau(t)||^2_f  +\gamma_3||\Delta Q^\text{Ad}_a(t)||^2_f+\gamma_4||S_1(t)||^2_f,	
		\label{eq:errfun1}
	\end{align} 
	which is parametric in $K_d(t)$  with $\gamma_i\in[0,1],~i=1,2,3,4$ for the interaction  \textbf{I1}, and 
	\begin{align} 
		\Theta_2(t)= \gamma_1||\Delta Q^\text{Ad}(t)||^2_f+\gamma_2||S_2(t)||^2_f+\gamma_3||S_3(t)||^2_f,
		\label{eq:errfun2}
	\end{align} 
	%\end{subequations} 
	which is parametric in $(K_d(t), K_\tau(t))$ with $\gamma_i\in[0,1],~i=1,2,3$ for the interaction \textbf{I2}. The gradient of $\Theta_1(t)$ with respect to $K_{d_i}(t)$, in interaction \textbf{I1}, is given by
	\begin{subequations} 
		\begin{align} 
			\nabla_{K_{d_i}(t)} \Theta_1(t)&=\left({\mathbf{e}_n^i}^\prime \otimes \idm{2}\right) ~\big[4\gamma_1R_d K_d(t) \mathcal I_d(t)  \Delta Q^\text{Ad}_d(t)  -4\gamma_2 {B}_d^\prime P_\tau(t) \Delta Q^\text{Ad}_\tau(t)  \notag  \\
			&\quad -4\gamma_3 {B}_d^\prime P_a(t)\Delta Q^\text{Ad}_a(t) 
			+2\gamma_4R_dS_1(t)\big]\mathcal I_{d_i}^\prime(t).
			\label{eq:errgradient1}
		\end{align} 
		Further, in the interaction \textbf{I2}, the gradient of $\Theta_2(t)$ with respect to $K_{d_i}(t)$ is given by
		\begin{align} 
			\nabla_{K_{d_i}(t)} \Theta_2(t)=&\left({\mathbf{e}^i_n}^\prime \otimes \idm{2} \right) \big[-4\gamma_1R_dK_d(t)\mathcal I_d(t) \Delta Q^\text{Ad}(t) 
			 +2\gamma_2R_dS_2(t)
			\big]\mathcal I_{d_i}^\prime(t),
			\label{eq:errgradient2}
		\end{align} 
		and with respect to $K_{\tau}(t)$ is given by
		\begin{align} 
			\nabla_{K_{\tau}(t)} \Theta_2(t)=\big[-4\gamma_1R_\tau K_\tau(t)\mathcal I_\tau(t)  \Delta Q^\text{Ad}(t) 
			+2\gamma_3R_\tau S_3(t)\big]\mathcal I_{\tau}^\prime(t). 
			\label{eq:errgradienttau}
		\end{align} 
		\label{eq:errgradient}
\end{subequations}}

 { 
\begin{theorem}  In interaction \textbf{I1}, let for every $t\in[0,T]$, $K_d^*(t)$ be the solution of the following optimization problem 
	\begin{align} 
	K^*_d(t)=\argmin_{K_d(t)} \Theta_1(t).\label{eq:optproblem}  
	\end{align}
	Then the NAFNE strategy parameterized by $K_d^*(t),t\in[0,T]$, that is,
	$\gamma_d^\text{Ad}(t,z(t);K_d^*(t)) $, $t\in[0,T]$ is
	a c-NAFNE strategy. Similarly, in interaction \textbf{I2}, 
	for every $t\in[0,T]$, $(K_d^*(t),K_\tau^*(t))$ be the solution of the following optimization problem
	\begin{align} 
	(K_d^*(t),K_\tau^*(t))=\argmin_{(K_d(t),K_\tau(t))} \Theta_2(t).\label{eq:optproblem2}  
	\end{align}
	Then the NAFNE strategy parameterized by $(K_d^*(t),K_\tau^*(t))$, $t\in[0,T]$, that is,\\
	$\gamma_{d\tau}^\text{Ad}(t,z(t);(K_d^*(t),K_\tau^*(t)))$, $t\in[0,T]$ is
	a c-NAFNE strategy.	 
	\label{thm:consistency}
\end{theorem}}
\begin{proof} 
  	 {For interaction \textbf{I1}, let $t_1\in[0,T]$ be a time instant in the game process when all the defenders can see all the players.
	From \eqref{eq:infomatrix}, this implies that the information matrices $\mathcal I_{d_i}(t_1)$ are non-singular for all $i=1,2,\cdots,n$.
	First, for the interaction given in P1, we show that the feedback gain matrix $\bar{K}_d(t_1)$ with its diagonal entries given by
	$\bar{K}_{d_i}(t_1)=-\left({\mathbf{e}^i_n}^\prime \otimes \idm{2} \right)R_d^{-1}B_d^\prime P_d(t_1)\mathcal I_{d_i}^{-1}(t_1) $ for $i=1,2,\cdots,n$ solves \eqref{eq:optproblem}.
	To see this, with the above choice of matrices the feedback gain matrix satisfies
	$\bar{K}_d(t_1)\mathcal I_d(t_1)=-R_d^{-1}B_d^\prime P_d(t_1)$.  Then using this in  \eqref{eq:perfindex}, gives
	$S_1(t_1)=0$, $Q_p^\text{Ad}(t_1)=Q_p$ for $p\in\{d,\tau,a\}$. Then, from \eqref{eq:errfun1} and \eqref{eq:errgradient1}, we get $\Theta_1(t_1)=0$ and $\nabla_{K_{d_i}(t_1)}\Theta_1(t_1)=0$ for all $i=1,2,\cdots,n$. This implies, $\bar{K}_d(t_1)$ minimizes $\Theta_1(t_1)$, that is, $K^*_d(t_1)=\bar{K}_d(t_1)$. Then, the control action at $t_1$ using the NAFNE strategy parameterized ${K}^*_d(t), t\in[0,T]$ satisfies $\gamma^\text{Ad}_d(t_1,z(t_1);K_d^*(t_1))={K}^*_d(t_1)\mathcal I_d(t_1)z(t_1)=-R_d^{-1}B_d^\prime P_d(t_1)z(t_1)=\gamma_d^*(t_1,z(t_1))$. This implies, from Definition \ref{def:consistency}, 
	$\gamma_d^\text{Ad}(t,z(t);K_d^*(t)),~t\in[0,T]$ is a c-NAFNE strategy.} 

 {In interaction \textbf{I2}, let $t_1\in[0,T]$ be the time instant when all the defenders and the target can see all the players. Then, from \eqref{eq:infomatrix} and \eqref{eq:infomatrixtau}  we have that the matrices $\mathcal I_d(t_1)$ and
	$\mathcal I_\tau(t_1)$ are invertible. Then using 	$\bar{K}_d(t_1)$ with its diagonal entries given by $\bar{K}_{d_i}(t_1)=\left({\mathbf{e}^i_n}^\prime \otimes \idm{2} \right)R_d^{-1}B_d^\prime P(t_1)\mathcal I_{d_i}^{-1}(t_1) $ for $i=1,\cdots,n$ and $\bar{K}_\tau(t_1)=R_\tau^{-1}B_\tau^\prime P(t_1)\mathcal I_{\tau}^{-1}(t_1) $, in \eqref{eq:perfmatzsg} we get 
	$S_2(t_1)=0$, $S_3(t_1)=0$, $Q^\text{Ad}(t_1)=Q$. Then, from \eqref{eq:errfun2}, \eqref{eq:errgradient2} and \eqref{eq:errgradienttau}, we get $\Theta_2(t_1)=0$, $\nabla_{K_{d_i}(t_1)}\Theta_2(t_1)=0$  for all $i=1,\cdots,n$ and $\nabla_{K_\tau(t_1)}\Theta_2(t_1)=0$.  This implies, $(\bar{K}_d(t_1),\bar{K}_\tau(t_1))$ minimize $\Theta_2(t_1)$, that is, $(K^*_d(t_1),K^*_\tau(t_1))=(\bar{K}_d(t_1),\bar{K}_\tau(t_1))$. Using the same arguments as before we have that 
	$\gamma_{d\tau}^\text{Ad}(t,z(t);(K_d^*(t),K_\tau^*(t))),~t\in[0,T]$ is a c-NAFNE strategy.} 
\end{proof} 		 
 {\begin{remark}  We note that for interaction \textbf{I1}, the optimization problem \eqref{eq:optproblem} is well-posed as 	$\Theta_1(t)\geq 0$ for all $K_d(t)$ and $t\in[0,T]$. From Theorem \ref{thm:consistency}, the gain matrices $K^*_d(t), t\in[0,T]$, obtained  from \eqref{eq:optproblem}, result in performance indices $(J_d^\text{Ad},J_\tau^\text{Ad},J_a^\text{Ad};K^*_d(.))$ which are closer to $(J_d,J_\tau,J_a)$. The performance indices  parameterized by
	the gain matrices $K_d^*(t),~t\in[0,T]$ can be referred to as best achievable performance indices; see \cite{lin2015nash} where this concept was introduced. For interaction \textbf{I2}, using similar arguments, it follows the performance index $(J^\text{Ad};(K^*_d(.),K^*_\tau(.)))$ parametrized by the gain matrices 
	$(K_d^*(t),K_\tau^*(t))$, $t\in[0,T]$ is the best achievable and closer to $J$. 
\end{remark}
In the next theorem, we study the effect of varying visibility radii on the c-NAFNE strategies.
\begin{theorem}
	Consider two TAD games with limited observations with identical problem parameters (including the initial state), and  differ only in defender $d_i$'s visibility radius in interaction \textbf{I1} (either defender $d_i$'s or target visibility radius in interaction \textbf{I2}).  Let $T_1$ and $T_2$ represent the time instants at which there exist an outgoing edge from the defender $d_i$ in interaction \textbf{I1} (either defender $d_i$'s or target in interaction \textbf{I2}) for the first time in these two games respectively. Then,  in interaction \textbf{I1} (\textbf{I2}) the control actions of the defenders (defenders and target) using their c-NAFNE strategies in these two games are identical during the time period $[0,\min\{T_1,T_2\})$. 
	\label{thm:delayproperty}
\end{theorem}
\begin{proof}
	 {In interaction \textbf{I1},	following the network feedback information structure, the defender $d_i$'s  information matrix \eqref{eq:infomatrix} satisfies $\mathcal I_{d_i}(t)=0$   for all $t\in [0,\min\{T_1,T_2\})$ in both the games. Further, as all other parameters in both the games are identical, except defender $d_i$'s visibility radius,  
		the joint equilibrium control actions of the defenders $u_d(t)$ using their c-NAFNE strategies is identical in both the games for all $t\in  [0,\min\{T_1,T_2\})$. Similar reasoning follows for a defender $d_i\in \mathcal D$ or the target in interaction \textbf{I2}.}
\end{proof}
\begin{remark} As the attacker and target use  their standard FNE strategies in interaction \textbf{I1}  their  state and control trajectories are also identical in these games during the time period $[0,\min\{T_1,T_2\})$; a similar 
	conclusion holds true only for the attacker in the interaction \textbf{I2}.
\end{remark}
\begin{remark}\label{rem:networkexternality} 
		The optimization problem, though well-posed, is non-convex and can be solved numerically.
		Further, from \eqref{eq:errgradient}  the computation of the gradient by the defender $d_i$ ($d_i\in \mathcal D$) in interaction \textbf{I1} (defender $d_i\in \mathcal D$ or target $\tau$ in interaction \textbf{I2}) requires joint feedback gain $K_d(t)$ (($K_d(t)$, $K_\tau(t)$)) and joint connectivity information $\mathcal I_d(t)$ ($\mathcal I_d(t)$, $\mathcal I_\tau(t)$).  In the real-world implementation, this information must be shared among the defenders (defenders and target) through a protocol as a part of cooperation. Such a protocol leads to a semi-decentralized implementation of c-NAFNE strategies.   
\end{remark} 
%%%%

%%%%
\section{Illustrative examples}	\label{sec:simulations} In this section, we illustrate the performance of c-NAFNE strategies studied in sections \ref{sec:optimization} through numerical experiments.   In real-world applications involving networked agents, with limited visibility, the presence of \emph{network externalities} plays an important role in the synthesis of players' strategies. In other words, when an outgoing link from a visibility constrained player forms or breaks, then it is important to know how this would affect the strategies of other team players who are not directly connected to this player. Besides verifying the obtained theoretical results, the numerical examples are designed to illustrate the
effect of network externalities.  To this end, we consider a $5$ player TAD game with $1$ target, $1$ attacker  and $3$ defenders. 
\subsubsection{\underline{ {Interaction \textbf{I1}}}} \label{sec:scenario1}  
 {Initially}, the  players $\{d_1,d_2,d_3,\tau,a\}$ are located at 
$\{(0,0),(1,1.5),(-1,0)$,\\ $(0,1),(-2,2)\}$ respectively. The control penalty parameter values are selected as $\{r_{d_1},r_{d_2},r_{d_3},r_\tau,r_a\}$\\ $=\{1,1,1,1.2,0.8\}$.
The interaction parameters  $\{q_{d_ia},f_{d_ia},q_{ad_i},f_{ad_i},q_{a\tau},f_{a\tau},q_{\tau a},f_{\tau a}\}$, $d_i\in \mathcal D$ in \eqref{eq:fullobj}  are set equal  to $1$, and the remaining  parameters are taken as  $T=6$,  {$\sigma_p=0.1$ for $p\in\{d_1,d_2,d_3,a\}$}. 
For implementation, we discretize the duration $[0,T]$ with a step size of $\delta=0.005$, and we use the matlab program \texttt{fminunc} for solving the  optimization \eqref{eq:optproblem} at each time step.
\begin{figure}[h] \centering \hspace{-0.15in}
	\begin{subfigure}[b]{.45\textwidth}\centering 
		\includegraphics[scale=1.5]{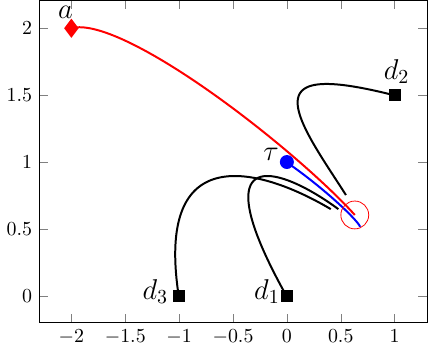}
		\caption{} 
		\label{fig:Scn12full}
	\end{subfigure}~~~
	\begin{subfigure}[b]{.45\textwidth}\centering 
		\includegraphics[scale=1.5]{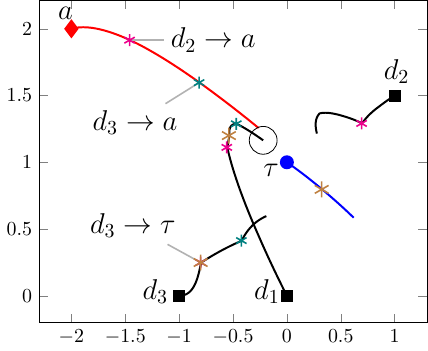}
		\caption{}
		\label{fig:Scn12pobs}
	\end{subfigure}  
\caption{Trajectories of players with complete (panel (a)) and limited  (panel (b)) observations with attacker operating in non-suicidal mode ($\lambda=1$). In panel (b) the markers illustrate the positions of the players at which the  labeled edge is active. }    
\end{figure}
%%%
\begin{figure}[h]  \centering 
	\begin{subfigure}[t]{0.3\textwidth}\centering 
		\includegraphics[scale=1.25]{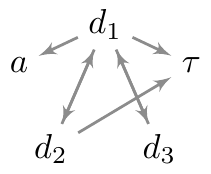}
		\caption{$t\in[0,0.16)$}
		\label{fig:Scn12gph1}
	\end{subfigure}~
	\begin{subfigure}[t]{0.3\textwidth}\centering 
		\includegraphics[scale=1.25]{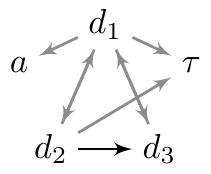}
		\caption{$t\in[0.16,0.49)$}
		\label{fig:Scn12gph2}
	\end{subfigure}~
	\begin{subfigure}[t]{0.3\textwidth}\centering 
		\includegraphics[scale=1.25]{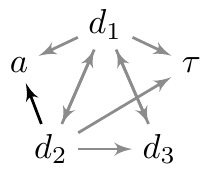}
		\caption{$t\in[0.49,0.575)$}
		\label{fig:Scn12gph3}
	\end{subfigure}		\vskip1.5ex
	\begin{subfigure}[t]{0.3\textwidth}\centering 
		\includegraphics[scale=1.25]{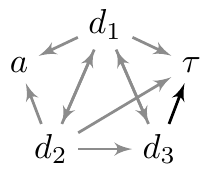}
		\caption{$t\in[0.575,0.915)$}
		\label{fig:Scn12gph4}
	\end{subfigure}	~
	\begin{subfigure}[t]{0.3\textwidth}\centering 
		\includegraphics[scale=1.25]{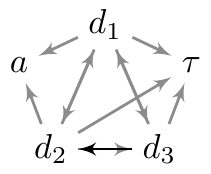}
		\caption{$t\in[0.915,0.96)$}
		\label{fig:Scn12gph5}
	\end{subfigure}	~
	\begin{subfigure}[t]{0.3\textwidth}\centering 
		\includegraphics[scale=1.25]{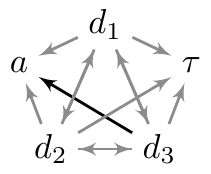}
		\caption{$t\in[0.96,1.395]$}
		\label{fig:Scn12gph6}
	\end{subfigure}	
	\caption{ {Evolution of visibility network for the interaction illustrated in Fig. \ref{fig:Scn12pobs}.  The dark arrow indicates the occurrence of a new edge in the network.}} 
\end{figure}
%%%
\begin{figure}[h]  \centering  \hspace{-0.15in}
	\begin{subfigure}[b]{.45\textwidth}\centering 
		\includegraphics[scale=1.5]{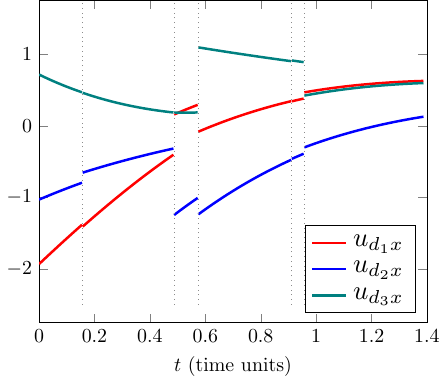}
		\caption{} 
		\label{fig:Scn12udx}
	\end{subfigure}~~~
	\begin{subfigure}[b]{0.45\textwidth}\centering  
		\includegraphics[scale=1.5]{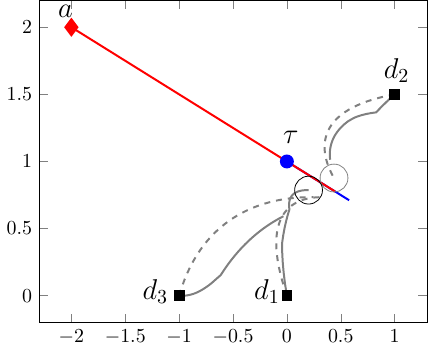} 
		\caption{}
		\label{fig:suicidalattacker}
	\end{subfigure}    
	\caption{{Panel (a) illustrates the c-NAFNE strategies of the defenders (along $x$-axis). The dotted vertical lines  in panel (a) illustrate time instants at which the visibility network changes structure.
			Panel (b) illustrates the trajectories of players with complete (dashed) and limited (solid line) observations with attacker operating in suicidal mode ($\lambda=0$).} } 
\end{figure}  
%%%
%%%
First we illustrate the scenario with a non-suicidal attacker. Fig. \ref{fig:Scn12full} illustrates the trajectories of the players with complete observations when all the players use their standard FNE strategies given by \eqref{eq:1d}. 
The game terminates at $t=2.66$ and results in the capture of the target by the attacker. Next, the visibility radii of the defenders are set to $\zeta_{d_1}=5$, $\zeta_{d_2}=2.25$ and $\zeta_{d_3}=1.25$. Fig. \ref{fig:Scn12pobs} illustrates the trajectories  using the c-NAFNE strategies from Theorem \ref{thm:thm1} and synthesized using Theorem \ref{thm:consistency}. The parameters in the optimization problem \eqref{eq:optproblem} are set as $\gamma_i=0.25,~i=1,2,3,4$ to indicate that the error terms in \eqref{eq:errfun1} are weighted equally.
Fig. \ref{fig:Scn12gph1}-\ref{fig:Scn12gph6} illustrate evolution of visibility network in the game. Whenever a new link forms (or disappears) in the network, and an additional gain term
is included (or deducted) from a defender's network adapted feedback strategy. \tb{Since all the defenders act as a team, this change in the network will reflect in all the defenders' control trajectories verifying the presence of network externalities; see also Remark \ref{rem:networkexternality}. Consequently, structural changes in the visibility network leads to jumps in the defenders' control trajectories; see Fig. \ref{fig:Scn12udx}. Further, these jumps lead to kinked state trajectories; see the labeled markers in Fig. \ref{fig:Scn12pobs}.}  
After $t=0.96$, all players can see each other, and from Theorem \ref{thm:consistency}, in the subgame starting at $t=0.96$ the c-NAFNE strategies coincide with the FNE strategies  given by \eqref{eq:1d}. \tb{This is because the c-NAFNE strategies synthesized using the information consistency criterion developed in section \ref{sec:optimization}}. 
At $t=1.395$, the game terminates with defender $d_1$ intercepting the attacker.  
Next, we set the parameter $\lambda=0$ to reflect the suicidal attacker while keeping all other parameters same as before. Fig. \ref{fig:suicidalattacker} illustrates the trajectories of the players with complete and limited observations. The game terminates with interception of the attacker by the defender $d_2$ ($d_1$) at time $t=2.375$ ($t=1.67$) with complete (limited) observations. \tb{We observe that the attacker and the target move along the straight line joining them at $t=0$, implying that the  strategies of the attacker and target are not affected by the visibility constraints of the defenders. These observations verify Theorem \ref{thm:straight_line}}.
%%%
%%%
%%%
\subsubsection{\underline{ {Interaction \textbf{I2}}}} 
\label{sec:scenario2}  	
 {Initially, the three defenders $d_1$, $d_2$, and $d_3$ are located at 
$(-1,0)$, $(-3,1)$, and $(1,2.5)$ respectively. The target  and the attacker are located at $(0.5,1)$ and $(-2.75,2.5)$ respectively. Except for the visibility radii of the defenders and target, the remaining parameters are set as in interaction \textbf{I1}. First, we analyze the effect of explicit cooperation of the defenders with the target under complete observations. Figure  \ref{fig:Int2full} illustrates the trajectories of the players when defenders act as a team against the attacker. The game ends at  $t=3.455$ with the interception of the attacker by the defender $d_1$. Next, we set the  visibility radii of the defenders as $\zeta_{d_1}=5 $ $\zeta_{d_2} =3 $ $\zeta_{d_3}=0.3$ and the target as $\zeta_\tau=10$.  Fig. \ref{fig:Int2part1}  illustrates the trajectories  using the c-NAFNE strategies from Theorem \ref{thm:thm2} and synthesized using Theorem \ref{thm:consistency}. The parameters in the optimization problem \eqref{eq:optproblem} are set as $\gamma_i=\tfrac{1}{3},~i=1,2,3$ to indicate that the error terms in \eqref{eq:errfun1} are weighted equally. The game terminates at $t=4.065$ with attacker capturing the target. Here, due to large visibility radius, the defender $d_1$ and the target $\tau$ can see all the players through out the game process. As observed in interaction \textbf{I1}, here also the changes in the network information influence other players in a team leading to kinks in the state trajectories.}
\begin{figure}[h] \centering \hspace{-0.15in}
	\begin{subfigure}[b]{0.45\textwidth}\centering 
		\includegraphics[scale=1.5]{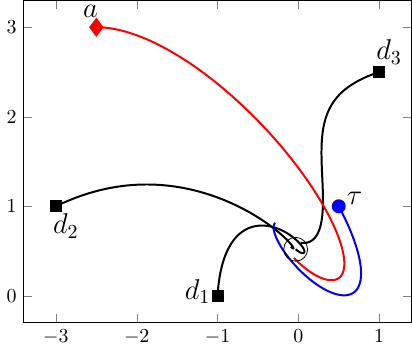}
		\caption{ } 
		\label{fig:Int2full}
	\end{subfigure}~~~
	\begin{subfigure}[b]{0.45\textwidth}\centering 
		\includegraphics[scale=1.5]{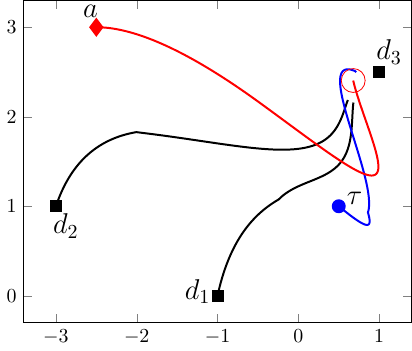}
		\caption{}
		\label{fig:Int2part1}
	\end{subfigure}  
\caption{{Panel (a) illustrates the trajectories of players with    complete observations.
		Panel (b) illustrates the trajectories of players with visibility constrained defenders-target team with parameters $\zeta_{d_1}=5,\zeta_{d_2}=3,\zeta_{d_3}=0.3, \zeta_\tau=10$.} }    
\end{figure}  
\begin{figure}[h] \centering \hspace{-0.15in}
	\begin{subfigure}[b]{0.45\textwidth}\centering 
		\includegraphics[scale=1.5]{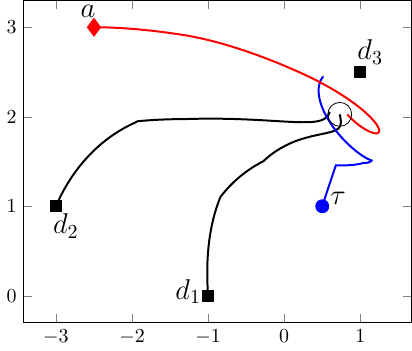}
		\caption{}
		\label{fig:Int2part2}
	\end{subfigure}~~~
	\begin{subfigure}[b]{0.45\textwidth}\centering 
		\includegraphics[scale=1.5]{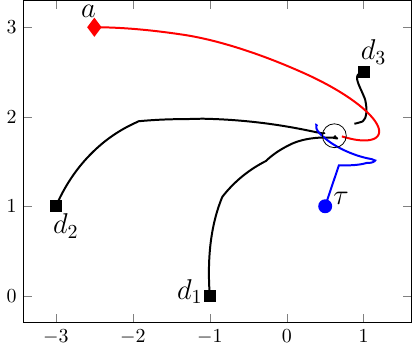}
		\caption{}
		\label{fig:Int2part3}
	\end{subfigure}  
	\caption{{Trajectories of players with visibility constrained defenders-target team, with parameters $\zeta_{d_1}=5,\zeta_{d_2}=3,\zeta_{d_3}=0.3, \zeta_\tau=2.5$ in panel (a) and 
		with parameters $\zeta_{d_1}=5,\zeta_{d_2}=3,\zeta_{d_3}=0.6, \zeta_\tau=2.5$ in panel (b).}}
\end{figure} 
 \\
  Next, we set the visibility radius of the target as $\zeta_\tau=2.5$, indicating that the target cannot see all the players initially.  Fig. \ref{fig:Int2part2} illustrates the trajectories of the players, and Fig. \ref{fig:Int2gph1}-\ref{fig:Int2gph5} illustrate evolution of the visibility information during this interaction. Fig. \ref{fig:Int2Control1} illustrate the c-NAFNE strategies (in the $y$-axis) of the defenders and the target. Notice, in the duration $[0,0.51)$ the target cannot see the attacker and its trajectory is influenced by $d_1$ (who can see the attacker) and $d_3$ (who cannot see the attacker). It moves towards the attacker in this period. At $t=0.51$, the target sees the attacker and start to move away from the attacker, resulting in the kink in its state trajectory (Fig. \ref{fig:Int2part2}) and jump in the control trajectory (Fig. \ref{fig:Int2Control1}). Here, defender $d_3$ cannot observe any player throughout the game process. However, the availability of the position information of the immobile $d_3$ influences the standard FNE strategies of the  attacker, and the  c-NAFNE strategies of other defenders and target who can observe $d_3$. \tb{Again, like in interaction \textbf{I1}, the effect of network externalities can be seen in the jumps in the c-NAFNE strategies of the team players; see Fig. \ref{fig:Int2gph1}-\ref{fig:Int2gph5} and Fig. \ref{fig:Int2Control1}}.
  The game ends at $t=3.46$ with defender $d_1$ intercepting the attacker.  
\begin{figure}[h]  \centering \hspace{-0.15in}
	\begin{subfigure}[t]{00.3\textwidth}\centering 
		\includegraphics[scale=1.25]{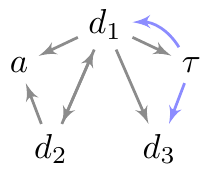}
		\caption{$t\in[0, 0.51)$}
		\label{fig:Int2gph1}
	\end{subfigure}~
	\begin{subfigure}[t]{00.3\textwidth}\centering 
		\includegraphics[scale=1.25]{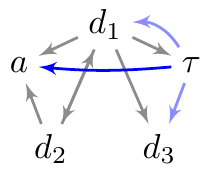}
		\caption{$t\in[0.51,   0.895)$}
		\label{fig:Int2gph2}
	\end{subfigure}~
	\begin{subfigure}[t]{00.3\textwidth}\centering 
		\includegraphics[scale=1.25]{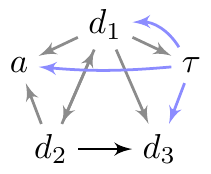}
		\caption{$t\in[0.895, 0.91)$}
		\label{fig:Int2gph3}
	\end{subfigure}	
	\vskip1.5ex 
	\begin{subfigure}[t]{00.3\textwidth}\centering 
		\includegraphics[scale=1.25]{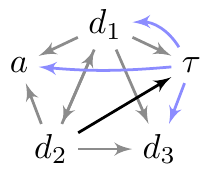}
		\caption{$t\in[0.91,  1.2 )$}
		\label{fig:Int2gph4}
	\end{subfigure}	~
	\begin{subfigure}[t]{00.3\textwidth}\centering 
		\includegraphics[scale=1.25]{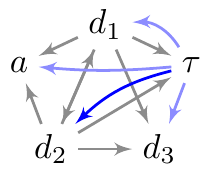}
		\caption{$t\in[1.2, 3.46]$}
		\label{fig:Int2gph5}
	\end{subfigure}	
\caption{ {Evolution of visibility network for the interaction illustrated in Fig. \ref{fig:Int2part2}}.} 
\end{figure} 
\begin{figure}[h] \centering \hspace{-0.15in}
		\begin{subfigure}[b]{0.45\textwidth}\centering 
			\includegraphics[scale=1.5]{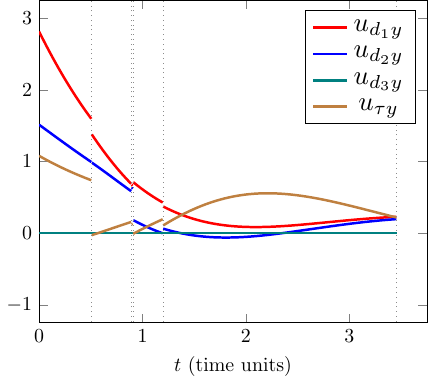}
			\caption{$  $}
			\label{fig:Int2Control1}
		\end{subfigure}~~~
		\begin{subfigure}[b]{0.45\textwidth}\centering 
			\includegraphics[scale=1.5]{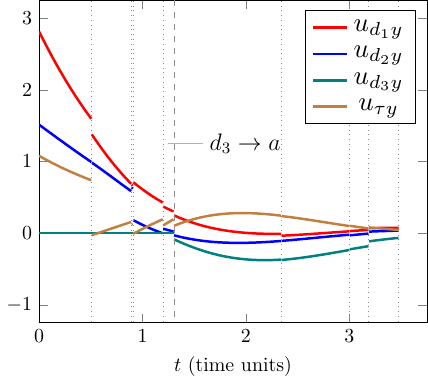}
			\caption{$  $}
			\label{fig:Int2Control2}
		\end{subfigure}  
		\caption{ {The visibility radii of other defenders and target set as $\zeta_{d_1}=5,\zeta_{d_2}=3,\zeta_\tau=2.5$. 
			Panel (a) illustrates the defenders-target team's c-NAFNE strategies (along $y$-axis) in the game with $\zeta_{d_3}=0.3$. Panel (b) illustrates the defenders-target team's c-NAFNE strategies (along $y$-axis) when $\zeta_{d_3}$ is changed from $0.3$ to $0.6$. The dashed vertical line at $t=1.305$ indicates that time instant when the edge $d_3\rightarrow a$ becomes active.}}
\end{figure} 
%%%
\\
To analyze the sensitivity of the visibility radii, we choose $\zeta_{d_3} =0.6$ keeping all the other parameters same. Fig. \ref{fig:Int2part3} shows the trajectories of the players. The game ends at $t=3.48$ with defender $d_1$ intercepting the attacker. We note that an outgoing edge  $d_3\rightarrow a$, from the defender $d_3$ becomes active for the first time at $t=1.305$. In the game with $\zeta_{d_3}=0.3$, 
the defender $d_3$ is inactive for the entire duration of the game. So, there does not exist any outgoing edge from the defender $d_3$ for the duration $[0,1.305)$ in both the games.   {From Fig. \ref{fig:Int2Control1} and \ref{fig:Int2Control2}, we notice that the c-NAFNE controls of the defenders and the target are identical in both the games for the duration $[0,1.305)$ and differ at and after $t=1.305$. This observation verifies Theorem \ref{thm:delayproperty}}.  
%%%
\section{Conclusions} 
\label{sec:conclusions}
In this paper, we have studied TAD games involving limited observations.   We have analyzed two variations leading to modeling the interactions as non-zero-sum and zero-sum differential games. We have demonstrated that the feedback strategies of the visibility constrained players must be adapted to the visibility network induced by these constraints, and introduced network feedback information structure. To obtain implementable strategies, we used an inverse game theory approach which leads to a plethora of  Nash equilibria, and we addressed this issue using an information consistency criterion. We have illustrated our results with numerical simulations.
The framework developed in this paper can be easily extended for analyzing other variations 
of TAD games involving limited observations, with an appropriate choice of game termination criterion.  In this paper, we have assumed that the visibility constrained players design their strategies based on what they can sense individually. However,  it is possible that teammates can communicate and share  information as a part of cooperation. Incorporation of these features would require development of a team centric consistency criterion, and we plan to investigate in this direction for future work. 
\bibliographystyle{IEEEtran}
\bibliography{TCNS}
\begin{IEEEbiography}[]{Sharad Kumar Singh} received the B.Tech. degree
	in electronics and communication engineering from Feroze Gandhi Institute of Engineering and Technology, Raebareli, India, in 2014. 
	He received M.S and Ph.D in  Electrical Engineering from the Indian Institute of Technology (IIT) Madras, Chennai, India, in 2022. Currently, he is working as a robotics researcher at Addverb Technologies, India. His research interests include mobile robotics, game theory, multiagent systems and optimal control.
\end{IEEEbiography}
\begin{IEEEbiography}[]{Puduru Viswanadha Reddy} received his Ph.D. degree in operations research from Tilburg University, The Netherlands, in 2011. He held post-doctoral appointment with GERAD, HEC-Montr\'{e}al, Montr\'{e}al, Canada, during 2012-16. He is currently an assistant professor at the Department of Electrical Engineering, Indian Institute of Technology -- Madras, Chennai, India. His current research interests are in game theory and optimal control and their applications in the control of multi-agent systems.
\end{IEEEbiography}

	\end{document}